\RequirePackage{algorithm}
\RequirePackage{algpseudocode}
\documentclass[11pt,a4paper]{article}
\usepackage[margin=1.3in]{geometry}
\usepackage[T1]{fontenc}
\usepackage{graphicx}
\usepackage{authblk}

\usepackage{hyperref}
\usepackage{amsmath}
\usepackage{amsfonts}
\usepackage{amsthm}
\usepackage{amssymb}
\usepackage{graphicx}
\usepackage{authblk}
\usepackage{float}
\restylefloat{table}
\usepackage{siunitx}
\usepackage{bbm}
\usepackage{bm}
\usepackage{comment}
\usepackage{xcolor}
\usepackage{cleveref}
\usepackage{dsfont}
\usepackage{braket}
\usepackage{tikz}
\usetikzlibrary{decorations.pathreplacing,arrows.meta,calc}
\usepackage{enumerate}
\usepackage{siunitx}
\usepackage{booktabs}
\usepackage{soul}
\usepackage [ lambda,advantage , operators , sets , adversary , landau , probability , notions , logic ,ff, mm,primitives , events , complexity , asymptotics , keys]{cryptocode}
\usepackage[dvipsnames]{xcolor}

\DeclareFontFamily{OT1}{mathc}{}
\DeclareFontShape{OT1}{mathc}{m}{it}{<-> mathc10}{}
\DeclareMathAlphabet{\mathabxcal}{OT1}{mathc}{m}{it}

\renewcommand{\poly}{\mathrm{poly}}

\newtheorem{proper}{Property}

\usepackage{color}

\newtheorem{theorem}{Theorem}

\newtheorem{corollary}{Corollary}

\begin{document}
	\title{Decentralized Consensus from Quantum Proof of Position}
	\author[1]{\large Dar Gilboa\thanks{darg@google.com}}
	\author[1,2]{\large Siddhartha Jain\thanks{sidjain@utexas.edu}}
	\author[3]{\large Or Sattath}
	\affil[1]{\small Google Quantum AI, Venice, CA, United States}
	\affil[2]{\small The University of Texas, Austin, TX, United States}
	\affil[3]{\small Ben-Gurion University, Beer Sheva, Israel}

	\maketitle
	
	\begin{abstract}
		We construct an energy efficient consensus protocol with geospatial position as an external resource by using quantum computers. Economic resources of some form are essential in decentralized consensus protocols in order to prevent attacks in which malicious parties costlessly generate fake identities.

		The resulting protocol, Quantum Proof of Position (QPoP), improves on existing classical mechanisms: it dispenses with the notorious electricity consumption of Proof-of-Work, and, because its resource is external to the protocol, with the wealth concentration of Proof-of-Stake.
        QPoP also has two properties that, to our knowledge, no other permissionless Sybil resistance mechanism possesses. First, its Sybil resistance is non-rival: position verification occupies a quantum computer only briefly, so a single device can secure any number of chains at once, each enjoying its full benefit. Second, the protocol generates geographic decentralization, because its resource cannot be agglomerated: a participant's power is the number of distinct locations at which it operates a quantum computer, and a second computer at an occupied location adds nothing. 

	\end{abstract}

	\section{Introduction}
	  
	The decentralized consensus problem has received considerable attention since the landmark Bitcoin protocol first demonstrated a viable solution in the fully permissionless setting~\cite{Nakamoto2008-bc}. Subsequent years have witnessed an explosion of protocols of various flavors~\cite{Buterin2014-nr, PassS17, Kiayias2017-is, Yakovenko2018-ot, Cohen2019-pn}, providing a base layer upon which various applications can be built~\cite{Roughgarden2024-wn, Buterin2022-yl}. While the original formulations of the consensus problem, often under the name of State Machine Replication (SMR), were in a setting where the set of participants is fixed in advance and known to all~\cite{Dolev1983-yg}, protocols designed to function in the decentralized setting of the web cannot make this assumption. The ideal degree of decentralization is the fully permissionless setting, where the set of participants is unknown to the protocol and may vary over time. Intermediate settings between the permissioned and fully permissionless setting are also common, most notably the quasi-permissionless setting where the protocol has some knowledge of the participants, though the set of active participants can vary over time (see~\cite{Lewis-Pye2023-zp} for a precise definition of these models, including intermediate ones).

		At a high level, any permissionless consensus protocol must address two questions: ``Who gets to speak?'' and ``When is a statement to be trusted/finalized?'' Addressing the former requires a mechanism of Sybil %
	resistance. Since the identity of participants is not fixed, this mechanism prevents a Sybil attack in which an adversary creates an arbitrary number of fake identities costlessly and hijacks the protocol. Sybil resistance is typically achieved by introducing a costly economic resource, and the protocol is designed such that the power of any participant in the protocol, honest or otherwise, is roughly proportional to the relative amount of this resource that they hold. 
	Some of the most common choices are computational power/hashrate (Proof-of-Work or PoW)~\cite{DN93,Nakamoto2008-bc} and stake in the cryptocurrency native to the blockchain (Proof-of-Stake or PoS)~\cite{Kiayias2017-is}, though several other choices are possible. Each leads to various tradeoffs in protocol design. %
	
	In parallel, ever since the advent of quantum money circa 1970, it has been known that the uncloneability of quantum states can have interesting cryptographic consequences. This has led to the construction of multiple quantum primitives such as quantum money and tokens~\cite{Wiesner1983-lt, Gavinsky12, Ben-David2023-in, Gavinsky2025-kb}, quantum key distribution~\cite{Bennett1984-xn}, copy-protected quantum software~\cite{Aaronson2009-uz} and one-shot signatures~\cite{Amos2020-ls, ShmueliZ25, Shmueli2025-oh} (see~\cite{Sattath2023-uc} for a review). 
	These can be used to create objects with tantalizing properties that are impossible to achieve classically without hardware assumptions, such as currency or software that cannot be copied or signatures that cannot be reused. Given the unique capabilities of quantum states to serve as an uncloneable resource, it is natural to ask whether better solutions to the decentralized consensus problem, and the Sybil resistance component in particular, exist in a quantum world. 
	
	\subsection{Quantum consensus}

    Indeed, quantum resources are known to enable improved solutions to the consensus problem. One notable example is a constant versus polynomial round complexity for byzantine agreement separations under a strong adversary model~\cite{Ben-Or2005-ag}. Another is the use of one-shot signatures~\cite{Amos2020-ls, ShmueliZ25} to improve the safety thresholds of consensus under partial synchrony~\cite{Drake2023-wu}, circumventing the well-known DLS lower bound~\cite{Dwork1988-dh}.  
    
	In addition, Zhandry gives a construction of a cryptocurrency based on quantum lightning~\cite{Zhandry2019-ru} (which can be built from one-shot signatures~\cite{ShmueliZ25}). His scheme uses PoW to mint new money, as in many other cryptocurrencies, but relies on the no-cloning property of quantum lightning to prevent double-spending. Overall, Zhandry’s approach eliminates the need for a blockchain or any form of consensus. The main drawback of Zhandry’s approach is that, without consensus, there is no way to regulate the rate at which money is minted and thus control inflation. This construction is closest in spirit to our work, yet it uses uncloneability to partially replace the consensus mechanism itself rather than the Sybil resistance mechanism (i.e., PoW). Put differently, this does not enable one to solve the basic question that consensus answers when applied to cryptocurrencies: “Where is the money?” 
	Ref.~\cite{CS20} piggybacks on an existing blockchain for Sybil resistance (where the coins are already allocated), and provides a mechanism to migrate to quantum money, thereby removing the need to further maintain the blockchain, while relying on no-cloning to prevent double-spends.   
	
	In this work, we attempt a fresh look at this question, and design a consensus protocol based on a novel quantum Sybil resistance mechanism. Our key idea is to use quantum computers located at distinct positions in space as the scarce resource, enabling Sybil resistance. We denote this mechanism \textit{Quantum Proof-of-Position} (QPoP). 
	Our protocol has several desired properties, most notably improved energy efficiency compared to Proof-of-Work mechanisms, and avoids the concentration of wealth that occurs in Proof-of-Stake protocols.
	A more detailed comparison to existing protocols is provided in \Cref{sec:comparison}. 
	
	While there are problems widely believed to be efficiently solvable only on a quantum computer, most notably factoring, it is not obvious how to use this fact to obtain improved Sybil resistance. This is because relying on computational tasks directly will likely require fine-grained analysis of the complexity of e.g. factoring in order to guarantee security, similar to the fine-grained complexity guarantees needed for classical Proof-of-Work~\cite{BallGHKP24}. The practical cost of factoring may be particularly sensitive to hardware improvements in a field that is progressing rapidly, and to tradeoffs between circuit depth and qubit counts which are common in quantum algorithm design (e.g., Regev's recent factoring algorithm~\cite{Regev2025-ho}). Even if the complexity of factoring or another problems exhibiting quantum advantage could be understood at the required resolution, it is still unclear what benefits would be conferred by replacing classical PoW with ``Proof of Quantum Work''. In our approach, there is a clear benefit in that the main function of the quantum computers is to (periodically) provide a proof of the quantum computer's position in space. This is known to be impossible classically~\cite{Chandran2009-ga}, but achievable using quantum computers~\cite{Buhrman2014-qx}, even when only using classical communication~\cite{Liu2022-ci}. 
	We discuss the problem of position verification and its quantum solution in greater detail in \Cref{sec:pv}. 
	
	One issue with using position verification for Sybil resistance is that the proof of position is not transferable, which precludes it from being used as a drop-in replacement for PoW in e.g., Nakamoto consensus. Fortunately, there is a class of hybrid protocols such as Hybrid~\cite{PassS17}, Byzcoin~\cite{KokorisKogiasJ16} and Solida~\cite{AbrahamMNRS17} that maintain a dynamic committee running a permissioned consensus protocol, and update it in a fair way using PoW. 
	The committee in these algorithms makes them a particularly good match for our resource, since only the committee members need to run position verification as the verifier. We show that QPoP can be substituted for Proof-of-Work in these protocols, focusing on Solida~\cite{AbrahamMNRS17}. We do this because Byzcoin does not come with formal guarantees (at least not at the time of original publication) and Solida is a simplification of Hybrid, though similar modifications of these other hybrid protocols might be possible. %
	 
	Another consequence of the non-transferability of position verification is that spam prevention in our setting is non-trivial. Since position verification requires communication and computation by the entire committee, one would like to avoid a scenario where an adversary costlessly publishes a large number of positions that must then be falsified. The use of quantum computers in our protocol suggests a natural solution. We present in \Cref{alg:pq_reg_rom} a registration algorithm that uses proofs of quantumness based on the discrete logarithm problem to ensure that registration to participate in the protocol is not costless, by relying on a Random Oracle assumption. The proof is publicly verifiable, using classical computation only. This can be seen as a highly coarse-grained, effectively binary form of Sybil resistance (distinguishing between possession of a sufficiently powerful quantum computer or lack thereof), which as mentioned above would appear insufficient on its own to build a useful Sybil resistance mechanism. We show, however, that it can in some sense be ``bootstrapped'' into a legitimate Sybil resistance mechanism based on position verification.
	
	The development of quantum computation will very likely impact classical consensus protocols due to the vulnerabilities of cryptographic schemes based on Elliptic Curve Cryptography, used by virtually every widely deployed blockchain, to quantum attacks. Due to their decentralized nature, blockchains are especially ill-suited to coping with such a threat, and transitioning to post-quantum cryptography is nontrivial. We do not concern ourselves with this problem in this work.

	\subsection{Comparison to other protocols} \label{sec:comparison}
	
	We summarize how QPoP compares to Proof-of-Work~\cite{Nakamoto2008-bc, AbrahamMNRS17, PassS17, KokorisKogiasJ16} and Proof-of-Stake~\cite{Buterin2014-nr, Kiayias2017-is}. Each point is discussed in full, with sources, in the correspondingly ordered part of \Cref{sec:notable_properties}.
	
	\begin{itemize}
		\item Properties that, to our knowledge, no other permissionless Sybil resistance mechanism possesses:
		\begin{itemize}
			\item \textit{Non-rivalry (\Cref{property:nonrival}):} Position verification occupies a quantum computer only briefly, so one device can secure any number of chains at once, each enjoying its full benefit. Hash rate and stake serve a single chain at a time, and merged mining and restaking share them only by coupling the chains involved.
			\item \textit{Geographic decentralization (\Cref{property:geographic}):} A participant's power is the number of distinct locations at which it operates a quantum computer, and a second device at an occupied location adds nothing. Hash rate fits in a warehouse and stake in a key; QPoP generates geographic decentralization rather than merely tolerating it.
		\end{itemize}
		\item Advantages over Proof-of-Work:
		\begin{itemize}
			\item \textit{Capital intensiveness (\Cref{property:capex}):} In equilibrium block rewards are dissipated on the cost of participation. Bitcoin's rewards are split between hardware and electricity; QPoP's rewards are spent almost entirely on quantum hardware.
			\item \textit{Energy efficiency (\Cref{sec:energy_efficiency}):} Since rewards flow into hardware rather than electricity, QPoP consumes drastically less electricity than PoW. 
			\item \textit{Security (\Cref{property:security}):} Participants are locked in by their capital rather than able to exit by switching off hardware, no mining gap arises, and an attacker cannot rent the resource.
			\item \textit{Deterministic finalization (\Cref{property:finalization}):} Transactions are finalized by a BFT committee and cannot be rolled back, unlike Bitcoin's probabilistic confirmation.
			\item \textit{No obsolescence (\Cref{property:obsolescence}):} Mining hardware is retired once its electricity costs exceed its rewards. QPoP has no such operating cost, so a registered device keeps contributing indefinitely.
		\end{itemize}
		\item Advantages over Proof-of-Stake:
		\begin{itemize}
			\item \textit{Equitability and censorship resistance (\Cref{property:equitability}):} In PoS the reward is the resource, so wealth compounds and leaders can censor to their own benefit. QPoP's resource is external, and its committee is sampled by position.
			\item \textit{Permissionlessness (\Cref{property:permissionless}):} Building a quantum computer and registering its position requires no other participant's approval, whereas acquiring stake requires a willing seller.
		\end{itemize}
	\end{itemize}
	
	Other external resources, such as disk space or elapsed time~\cite{Dziembowski2015-cu, Bowman2021-mr, Cohen2019-pn}, typically rely on non-standard primitives such as Verifiable Delay Functions~\cite{Boneh2018-jq}, which cannot be built in the Random Oracle model~\cite{ZRW25}.
	
	\begin{table}[t]
		\label{tab:comparison}
		\begin{center}
			\begin{tabular}{l l l l l}
				\toprule[.5mm]
				\textbf{Protocol} & \textbf{\vtop{\hbox{\strut Non-}\hbox{\strut rivalry}}} & \textbf{\vtop{\hbox{\strut Geographic}\hbox{\strut decentralization}}} & \textbf{\vtop{\hbox{\strut Energy}\hbox{\strut cost}}} & \textbf{Equitability} \\  
				\midrule
				Bitcoin (PoW)& {\bf \color{BrickRed}No} & {\bf \color{BrickRed}No} & {\bf \color{BrickRed}High} & {\bf \color{ForestGreen}Possible} \\ 
				Ethereum (PoS)& {\bf \color{BrickRed}No} & {\bf \color{BrickRed}No} & {\bf \color{ForestGreen}Low} & {\bf \color{BrickRed}Impossible} \\
				Solida (PoW)& {\bf \color{BrickRed}No} & {\bf \color{BrickRed}No} & {\bf \color{BrickRed}High} & {\bf \color{ForestGreen}Possible} \\ 
				\textbf{QPoP} & {\bf \color{ForestGreen}Yes} & {\bf \color{ForestGreen}Yes} & {\bf \color{ForestGreen}Low} & {\bf \color{ForestGreen}Possible} \\
				\bottomrule[.5mm]
			\end{tabular}
			\caption{A comparison of our protocol with prominent blockchains. QPoP uses quantum position verification as an off-chain resource. This resource is non-rival and cannot be agglomerated in one location, neither of which holds for hash rate or stake, and it delivers improved energy efficiency compared to PoW protocols while avoiding the wealth concentration faced by protocols that rely solely on on-chain resources such as stake.
				}
		\end{center}
	\end{table}

	The details of the protocol are provided in \Cref{sec:details}, and in \Cref{sec:consistency_liveness} we show that it satisfies consistency, liveness and avoids spam. We survey results on the quantum resources required to run certain subroutines used by the protocol in \Cref{sec:resource_requirements}, and discuss various open questions in \Cref{sec:discussion}.
	
	\section{Preliminaries}
	\subsection{Position verification} \label{sec:pv}
	
	Position-based cryptography~\cite{Chandran2009-ga} is motivated by the idea that identity in a cryptographic setting can be established based on physical location. The original work that formalized this concept showed that secure position verification is impossible when only classical resources are used. A position verification protocol involves a prover and a set of verifiers, and it was shown in~\cite{Chandran2009-ga} that for any classical protocol, the part of the prover can be simulated by an adversary that is not present at the claimed position. Interestingly, it was later shown that this is no longer the case if provers and verifiers are equipped with quantum computers~\cite{Buhrman2014-qx,Unruh2014-eb,BCS22}. At a high level, the classical impossibility results can be evaded in a quantum world because the attack in~\cite{Chandran2009-ga} requires that the adversary copies messages sent by the verifiers to the prover. If the messages are quantum states that are unknown to the adversary, they can no longer be copied by the no-cloning theorem~\cite{WZ82}. While the original proposals for quantum position verification required quantum communication, this requirement was later removed by combining position verification protocols with proofs of quantumness~\cite{Liu2022-ci}. 
	
	Proofs of quantumness~\cite{Mahadev2022-ts, Brakerski2021-bg, Brakerski2023-ok, Vidick2019-qm} are powerful tools that allow a classical verifier to certify that a prover holds a quantum computer, and can in fact be used to certify any efficient quantum computation. The original constructions of proof of quantumness rely on families of Noisy Trapdoor Claw-Free Functions (NTCFs), which map pairs of inputs to a single output, and can only be inverted efficiently given a secret trapdoor. A quantum computer can be used to prepare a superposition over the inputs of such a function that map to any output, and solve a challenge that requires such a state as input. A classical algorithm capable of solving this challenge could be rewound and used to recover both inputs of the NTCF that map to the same output, breaking the hardness assumption associated with the NTCF. One natural construction is based on hardness of Learning With Errors (LWE)~\cite{Regev2010-zq}. The idea of~\cite{Liu2022-ci} is to have a \emph{classical verifier} by replacing the quantum messages in position verification with a (classical) Proof-of-Quantumness. There are certain timing constraints that must be satisfied so that the prover cannot cheat, but the authors show that this enables secure position verification using only classical messages. As with other position verification schemes, adversaries can spoof their position if they pre-share entanglement, yet the authors also show how to render position verification secure in such settings. We summarize their main results in \Cref{tab:cvpv_results}. 
	\begin{table}[ht]
		\centering
		\caption{Security of Classically Verifiable Position Verification (CVPV)~\cite{Liu2022-ci}. Here polynomial refers to $n^c$ for all $c$, whereas fixed polynomial refers to $O(n^c)$ for some fixed constant $c$. The most relevant result for use is Theorem 1.1, since it provides security in the standard model under a reasonably powerful adversary model.}
		\label{tab:cvpv_results}
		\begin{tabular}{l l l l l}
			\toprule[.5mm]
			\textbf{Theorem} & \textbf{\vtop{\hbox{\strut Quantum Hardness}\hbox{\strut Assumption}}} & \textbf{\vtop{\hbox{\strut Adversary}\hbox{\strut Entanglement}}} & \textbf{\vtop{\hbox{\strut Adversary}\hbox{\strut Complexity}}} & \textbf{Model} \\ 
			\midrule
			\textbf{1.1} & LWE (Polynomial) & None & Polynomial & Standard \\ 
			\textbf{1.2} & LWE (Subexponential) & Fixed polynomial & Subexponential & Standard \\ 
			\textbf{1.3} & LWE (Polynomial) & Polynomial  & Polynomial & QROM~\cite{Boneh2011-dm} \\
			\bottomrule[.5mm]
		\end{tabular}
	\end{table}
	\subsection{Notation and terminology}
	
	We use $c,C$ for (usually small/large resp.) absolute constants whose value can change from line to line. We denote by $a \mathbin\Vert b$ the concatenation of strings $a$ and $b$. We use Python conventions for dictionaries and arrays, namely $\textsf{d}.\textsf{keys}$ is the set of keys of a dictionary \textsf{d}, the value stored with key $\textsf{k}$ is given by $\textsf{d}[\textsf{k}]$, and $[a, \dots, b]$ is an ordered list. We denote our security parameter by $\lambda$ and define negligible probability with respect to it.
	
	When we refer to quantum computers, we mean devices that are capable of running fault-tolerant quantum programs at input sizes that are cryptographically relevant (in particular, we will use Shor's algorithm for the discrete logarithm problem~\cite{Shor1994-hb} and Mahadev's proofs of quantumness based on hardness of LWE~\cite{Mahadev2022-ts, Regev2010-zq}). In practice, this will likely require access to an order of $10^3$ error-corrected qubits. At any stage of the protocol where participants are associated with public keys, we will refer to entities holding such keys as nodes (with honest nodes holding a unique signing key associated with a public key, and an adversary potentially holding multiple keys). We will refer to an entity in possession of an off-chain resource (such as a quantum computer) as a node, even though such a participant is effectively anonymous from the perspective of the protocol, and is not yet associated with a public key. This participant has the freedom to generate new digital signature keys for every computer it wishes to register (or for the same computer if it is registered more than once). This is analogous to the use of fresh public keys in fully permissionless protocols. The intended meaning should be clear from context.

	\section{Hybrid Consensus with Quantum Position Verification} \label{sec:details}
	
	The proposed protocol shares most of the ingredients of Solida. Recall that Solida runs Byzantine Agreement among a set of Committe Members, specified by their (digital-signature) public keys. The size of the committee is $n$, where the genesis committee is hard-wired as $\mathcal{C}_1=[\textsf{pk}_1,\ldots, \textsf{pk}_n]$. The committee has two main roles: (a) approving batches of transactions, and (b) performing \emph{reconfiguration}, in which a new member is elected in a fair manner (based on an off-chain resource which prevents Sybil attacks and censorship), and the oldest member leaves. 
	
	Role (a) above uses standard techniques~\cite{Castro1999-tp}: the committee decides on a leader, who proposes the batch of transactions to approve. A malicious (Byzantine) leader cannot violate the safety of transactions (such as by double-spending), but \emph{can} violate progress. Thus, if there is no progress, eventually the committee elects a new leader, in a round-robin fashion from the current committee.
	
	The main challenge is role (b). Solida's approach is to use Proof-of-Work to determine the next committee member: a reconfiguration event occurs whenever a Proof-of-Work puzzle is solved. This allows users to join the committee in an unbiased fashion: the probability of a user joining the committee is proportional\footnote{Up to corrections due to network delays.} to their hashing power. In the terminology of Lewis-Pye and Roughgarden~\cite{Lewis-Pye2023-zp}, Solida and other hybrid protocols use both \textit{on-chain} (committee membership) and \textit{off-chain} (computational power) resources.
	
	This work uses a different reconfiguration rule and off-chain resource. The protocol defines a partition $\mathcal{P}$ of the physical space that participants can occupy. It also maintains a dictionary $\mathcal E$ of \emph{eligible} positions and public keys registered at those positions. Specifically, for every registered $\textsf{pos} \in \mathcal{P},$ 
	\begin{equation}
		\mathcal{E}[\textsf{pos}]=[\textsf{pk}_{1},\dots,\textsf{pk}_{k}],
	\end{equation}
	where $[\textsf{pk}_{1},\dots,\textsf{pk}_{k}]$ is a list of public keys that were registered at position $\textsf{pos}$ in the current reconfiguration step, in lexicographic order. Any participant can publish their position and public key in order to be included in $\mathcal{E}$. Position registration itself is fully permissionless. To prevent spam at this stage, there is an option to require a proof of quantumness by solving a discrete logarithm puzzle. 
	
	Reconfiguration is then performed by sampling a position uniformly at random from $\mathcal{E}.\textsf{keys}$, using a randomness beacon. An eligible committee member \emph{must} have a quantum computer. The main challenge is double-eligibility: making sure that a single computer would not allow a (dishonest) user to be counted as eligible twice. To prevent double eligibility, we employ \emph{position verification}, and for improved practicality, we utilize a classically verifiable position verification (CVPV) protocol~\cite{Liu2022-ci}. Note that a user with one quantum computer can in principle register two different locations, by registering once, moving the quantum computer to a different cell in the partition $\mathcal{P}$ and registering again at the second location (which must be at least a distance $\Gamma$ from the first). This user would not benefit from an increased probability of being elected to the committee because they can only successfully run CVPV as the prover from their current position. %
	The committee then runs CVPV, assuming the prover is at the sampled position from $\mathcal{E}.\textsf{keys}$, verifying that the messages from the prover are signed appropriately. By security of CVPV, the prover will succeed only if they possess a sufficiently powerful quantum computer at that position. This mechanism can be used to construct a novel form of Sybil resistance, with favorable properties compared to classical alternatives.
	
	Next, we make our formulation more precise. Our protocol will require the following assumptions:
	\begin{enumerate}
		
		\item \textit{Trusted setup:} The genesis committee $\mathcal{C}_1$ has at most $(1/3-\varepsilon)n$ Byzantine nodes for some constant $\varepsilon$, and genesis commitee members are assigned a random ordering. 
		\item \textit{CVPV resolution:} There is a fixed partition of space $\mathcal{P}$ into cells of size $\Gamma$. A lower bound on the possible size of the cells in the partition will be set by the spatial resolution of CVPV, and it would be beneficial to set this as small as possible.
		\item \textit{Cell availability:} An honest node can place a quantum computer in a cell's center, and prevent others from placing their computer in that cell.  
		
		\item \textit{Randomness beacon:} We make this assumption for convenience since it is orthogonal to Sybil resistance which is our main concern, as in e.g.~\cite{Shi2020-aw, BallGHKP24}. Decentralized randomness beacons are generally a challenge for any protocol that does not rely on Proof-of-Work~\cite{Bonneau2022-di}.
		\item \textit{Communication:} i) Synchronous %
		communication between nodes on the committee during $\textsf{SteadyState}$ and $\textsf{ViewChange}$, with maximal message delay $\Delta$. ii) Sufficiently low latency communication between verifiers and provers in order to execute position verification in $\textsf{Reconfiguration}$, which could in practice be achieved using cellular towers or satellite communication. iii) All participants (whether on the committee or not) can communicate using a gossip protocol~\cite{Demers1987-rr, Jelasity2007-zy}, but without any guarantee of consistency. Assumption ii) is non-standard and is a weakness of our protocol, yet appears to be an unavoidable consequence of the demands of position verification and the novel capabilities it provides. It is only required between a small number of participants at any given time. 
		\item Assumptions required for CVPV security (see \Cref{sec:pv}). In particular, we assume quantum hardness of LWE. Security against adversaries pre-sharing unbounded entanglement requires the QROM, but as long as entanglement is polynomially bounded CVPV is secure in the standard model. It is also worth noting that sharing entanglement across long distances that can then be used for computation is an extremely challenging engineering problem, usually requiring transduction between optical and memory qubits~\cite{Lauk2020-du}.
		\item \label{ass:bounded_qcomp} All participants (byzantine or honest) have polynomially bounded computational power (both quantum and/or classical).
		\item Byzantine nodes control at most a fraction $\rho$ of quantum computers at any given time slot. Security of the protocol is guaranteed when $\rho < 1/3-\varepsilon$.
		\item \label{ass:rom} (For spam resistance) Random Oracle. This assumption is used only to mitigate spam, and not for the security of the position verification-based Sybil resistance mechanism.
		\item \label{ass:dlog}(For spam resistance) \textit{DLP assumption in $\mathcal{QR}_p$:} Classical average-case hardness of the discrete logarithm problem in the subgroup of quadratic residues \cite[Definition 10.6]{BS23}. More precisely, let $p=2q+1$ where $p$ and $q$ are primes ($p$ is often called a safe prime), let $\mathcal{QR}_p \subseteq \mathbb Z_p^*$ be the quadratic residues subgroup of order $q$, and $g$ a generator of $\mathcal{QR}_p$. For $h\in \mathcal{QR}_p$, the discrete logarithm $\log_g(h)$ is defined as the unique $x \in \mathbb Z_q$ such that $g^x\equiv h \pmod p$. We assume that for a uniformly random\footnote{Equivalently, and as we will do later, $h$ may be sampled as $h := t^2 \pmod p$ for uniform $t\leftarrow \mathbb{Z}_p^*$.} $h \in \mathcal{QR}_p$ , given $p,g,h$, computing the discrete logarithm $\log_g(h)$ cannot be done classically by a probabilistic polynomial time algorithm, except with a success probability that is negligible in the security parameter. We denote by $R_H(\lambda)$ and $R_A(\lambda)$ a lower and upper bound respectively on the rate at which an honest participant and an adversary can solve the discrete log problem with security parameter $\lambda$ on a quantum computer. 
		\item The cost of running a quantum computer, whether for solving the discrete log or CVPV, significantly outweighs the cost of running CVPV as a verifier by the entire committee (which requires only classical computation and communication) and running PBFT. This is reasonable given the significant overheads of quantum error correction, which are fundamental in leading hardware platforms such as superconducting qubits due to their local connectivity constraints~\cite{Babbush2021-nu}. 
	\end{enumerate}

	The operation of the protocol can be divided into distinct stages:
	\begin{itemize}
		\item $\textsf{SteadyState}:$ The ``normal'' operation mode in which blocks of transactions are proposed by a leader on the committee and other members vote in order to finalize the blocks.
		\item $\textsf{ViewChange}:$ If a Byzantine leader does not propose blocks within a time frame (ultimately set by the maximal message delay,  $\Delta$), the committee replaces the leader by initiating a view change.
		\item $\textsf{Reconfiguration}:$ The committee is modified by adding a member that has proved it holds some external resource in order to guarantee Sybil resistance. An existing member is removed in the process.
	\end{itemize}
	
	$\textsf{SteadyState}$ and $\textsf{ViewChange}$ depend solely on on-chain resources so we leave them unchanged from the Solida protocol~\cite{AbrahamMNRS17}. For completeness, we include a description of these parts of the algorithm in \Cref{app:solida_algs}. $\textsf{Reconfiguration}$ in Solida is based on PoW (and is consequently event-driven). We replace it by a procedure based on CVPV that occurs every fixed number of time steps.
	
	The basic protocol is presented in \Cref{alg:hcqpv}. Nodes in $\mathcal{C}_i$ are identified by their public keys. Other participants specified in the algorithm are (honest) entities that can control quantum computers at some point in time during the execution. Stages in the algorithm which require interaction are assumed to be run for a time that allows an honest player to participate (for example, when registration requires computing the discrete log, we assume the time between registration periods is sufficient to compute on a quantum computer).  
	
	All steps below, unless specified otherwise (in a right-justified comment), can be performed by any party wishing to keep up with the protocol, and rely only on public information. Any classical party can efficiently verify the state of the protocol in this way, and for example read off the current makeup of the committee $\mathcal{C}_i$.

		\begin{algorithm}[htbp]
			\caption{Hybrid consensus with Quantum Proof-of-Position} 
			\label{alg:hcqpv}
			\begin{algorithmic}[1]
				\State Initialize committee $\mathcal{C}_1$, eligible position dictionary $\mathcal{E} = \emptyset$, $t=1$, $i=1$.
				
				\While{True}
				\If{$t = 0 \mod \tau_{\text{reconfig}}$}
				\State Update $\mathcal{E}$ using Location Registration (\Cref{alg:pq_reg_rom} or \Cref{alg:pq_reg_rom_spam}). 
				\While{$|\mathcal{E}|>0$} \label{step:cand_elec}
				\State Use randomness beacon to sample candidate $\mathsf{pos}$ from $\mathcal{E}.\textsf{keys}$. \label{step:sample_cand} 
				\State Set $(\mathsf{cpos}, \mathsf{cpks}) \leftarrow (\mathsf{pos}, \mathcal{E}[\mathsf{pos}])$
				\State Run Position Verification (\Cref{alg:pv_comm}) with $(\mathsf{cpos},\mathsf{cpks})$. \Comment{$\mathcal{C}_i$ and prover.}
				\If{Position Verification (\Cref{alg:pv_comm}) returns Success}
				\State \textbf{break} \Comment{Committee has been updated.}
				\Else
				\State Remove $\mathsf{pos}$ from $\mathcal{E}.\textsf{keys}$.  \Comment{No active quantum computer at $\mathsf{pos}$.}
				\EndIf
				\EndWhile \label{step:cand_elec_end}
				
				\Else
				\State Run $\textsf{SteadyState}$ for $T'$ slots (and $\textsf{ViewChange}$ as needed). \Comment{$\mathcal{C}_i$.}
				\EndIf
				\State Increment $t$.
				\EndWhile

			\end{algorithmic}
		\end{algorithm}
	
		\begin{algorithm}[htbp]
			\caption{Position Verification with Committee $\mathcal{C}$ and claimed position and keys $(\mathsf{pos}, \mathsf{pks})$} 
			\label{alg:pv_comm}
			\begin{algorithmic}[1]
				\For{$\textsf{pk} \in \textsf{pks}$}
				\For{$j \in \mathcal{C}$} \Comment{(Between times $t_0 + (j - 1) \tau_v$ and $t_0 + j \tau_v$)}
				\State Run CVPV as verifier, assuming prover at $\mathsf{pos}$. \Comment{j.}
				\State \label{step:cvpv_proof}Run CVPV as prover, signing messages with $\textsf{pk}$. \Comment{(Honest) prover.}
				\If{CVPV succeeds and all messages signed by $\textsf{pk}$} \Comment{j.}
				\State Set $r_{j}=1$. 
				\Else
				\State Set $r_{j}=0$.
				\EndIf
				\EndFor
				\State Run Byzantine Agreement with inputs $r_{j}$, returning $v_j$. \Comment{$\mathcal{C}.$}
				\State Each committee member publishes $v_j$. \Comment{$\mathcal{C}.$}
				\If{At least $2n/3$ members of $\mathcal{C}_i$ published $1$} \label{step:publish_results}
				\State Add $(\mathsf{pk}, \mathsf{pos})$ to $\mathcal{C}_i$, kick out most senior member.
				\State Increment $i$.
				\State \textbf{Return} Success
				\EndIf
				\EndFor
				\State \textbf{Return} Failure
			\end{algorithmic}
			
		\end{algorithm}
	
	Processing $\textsf{pks}$ one at a time in some arbitrary order ensures that all committee members are considering the same candidate at each stage. The purpose of publishing the results in \Cref{step:publish_results} is to enable other participants not on the committee to keep track of the state of the committee. Since the published $v_j$ are the outputs of Byzantine Agreement, all honest members will publish the same value, and hence a simple quorum of $2/3$ provides sufficient evidence for any external party that CVPV succeeded or failed. The CVPV steps refer to running either Construction 5.10 or Construction 6.2 from~\cite{Liu2022-ci}, depending on whether the adversary is assumed to possess an arbitrary amount of pre-shared entanglement or not.

		\begin{algorithm}[htbp]
			\caption{Spam-Resistant Location registration 
			}
			\label{alg:pq_reg_rom}
			\begin{algorithmic}[1]
				\State A public parameter $p=2q+1$ for $p,q$ primes ($p$ is often called a safe prime), and a generator $g$ for the (cyclic) subgroup of quadratic residues of $\mathbb Z_p^*$. 
				\State Use randomness beacon to sample a string $r$. Start timer $t'=0$.
				\While{$t' < \tau_{\text{register}}:$}
				\State \label{step:dlog_poq} Publish $(\textsf{pk},\textsf{pos},\log_g(H(\textsf{pk} \mathbin\Vert \textsf{pos} \mathbin\Vert r)^2 \pmod p))$ \Comment{Any party.}
				\EndWhile
				\State Wait until $t'=\tau_{\text{register}}+\Delta$.
				\For{Each message $(\textsf{pk},\textsf{pos},x)$ received}
				\If{$g^x\equiv H(\textsf{pk} \mathbin\Vert \textsf{pos} \mathbin\Vert r)^2 \pmod p $}
				\State Append $\textsf{pk}$ to $\mathcal{E}[\textsf{pos}]$, maintaining order.
				\EndIf
				\EndFor
			\end{algorithmic}
		\end{algorithm}
        
    		\begin{algorithm}[htbp]
			\caption{Location registration 
			}
			\label{alg:pq_reg_rom_spam}
			\begin{algorithmic}[1]
				\setcounter{ALG@line}{0}
				\State Start timer $t'=0$.
				\While{$t' < \tau_{\text{register}}:$}
				\State Publish $(\textsf{pk},\textsf{pos})$. \Comment{Any party.}
				\EndWhile
				\State Wait until $t'=\tau_{\text{register}}+\Delta$.
				\For{Each message received}
				\State Append $\textsf{pk}$ to $\mathcal{E}[\textsf{pos}]$, maintaining order.
				\EndFor
			\end{algorithmic}
		\end{algorithm}
	
	In \Cref{alg:pq_reg_rom}, we denote by $H$ a public hash function, modeled as a Random Oracle, and by $\log_g(h)$ the discrete logarithm of $h$ (see the DLP assumption for more details on p.~\pageref{ass:dlog}). 
	
	Note that the steps in \Cref{alg:pq_reg_rom} or \Cref{alg:pq_reg_rom_spam} that update $\mathcal{E}$ can be performed using public information by any participant. The ordering of $\mathcal{E}.\textsf{keys}$ is arbitrary but fixed, and therefore public, which renders the sampling in \Cref{step:sample_cand} in \Cref{alg:hcqpv} unbiased and independent between reconfigurations. The random string $r$ is chosen such that the DLOG problem is secure with security parameter $\kappa$.
	
	\subsection{Variants and Optimizations}
	
	The algorithms as presented are aimed at maximizing energy efficiency. Another variant that we do not consider explicitly is one where registration is not persistent, and $\mathcal{E}$ is discarded after each reconfiguration. This requires honest participants to register even if they have not been selected, and would lead to quantum computation scaling like $NT$ for $N$ participants. Since the DLOG puzzle will be efficiently solvable using a quantum computer, the overall energy requirements might still be modest compared to protocols based on Proof-of-Work, yet would now reduce to constant factors. The benefit of this version is that the adversary cannot bloat $\mathcal{E}$ with fake registration messages, and we would only require $\lambda > \log(R_A(\lambda)/R_H(\lambda))$ instead of $\lambda > \log(T R_A(\lambda)/R_H(\lambda))$ in order to guarantee spam-resistance. 
	
\section{Consistency, Liveness, and Spam Prevention} \label{sec:consistency_liveness}
	
	We prove that our protocol satisfies the following:
	\begin{proper}[Consistency and Liveness] \label{prop:consistency_liveness}
		Except with negligible probability, up to some $T=\poly(\lambda)$:
		\begin{enumerate}
			\item Consistency: 
			\begin{itemize}
				\item No rollbacks: If a transaction $\textsf{tr}$ is confirmed by participant $p$ at some time $t$, it remains finalized for every $t'>t$\footnote{Confirmation can be formalized as a function that maps a set of messages to a subset of messages that are valid w.r.t. the initial resource distribution (in our case the Genesis committee $\mathcal{C}_1$), as in ~\cite{Lewis-Pye2023-zp}. The longest-chain confirmation rule of Bitcoin for example, is that a transaction is sufficiently deep on the longest chain.}.
				\item No conflicts: If $T,T'$ are confirmed transactions for honest $p,p'$ at times $t,t'$, then $T \cup T'$ is a valid set of transactions with respect to the initial committee $\mathcal{C}_1$.
			\end{itemize}
			\item Liveness: A valid transaction received by an honest player at some time $t$ is eventually (after time $t+\poly(n, \Delta,K)$, where $K$ is the maximum number of position registrations) confirmed by every other honest node that is active at that time. 
		\end{enumerate}
	\end{proper}

	Our protocol enforces Sybil resistance based on position verification. Specifically, it possesses the following property:
	
	\begin{proper}[Sybil resistance from Quantum-Proof-of-Position] \label{claim:syb}
		Given $N$ participants each possessing $\mu_i$ quantum computers at distinct cells in the partition $\mathcal{P}$ that are online during the $j$th reconfiguration step,
		\begin{equation}
			\mathbb{P}\left[\text{participant }i\text{ is added to } \mathcal{C}_j\right]=\frac{\mu_{i}}{\underset{k=1}{\overset{N}{\sum}}\mu_{k}},
		\end{equation}
		and moreover these events are independent between reconfiguration rounds\footnote{This assumes all parties are using the computers in their possession to participate in the protocol, including Byzantine nodes, which is the worst-case scenario from the perspective of honest nodes. If this were not the case, the equality could be replaced with an $\geq$.}.
	\end{proper}
	In proving liveness of our protocol, we will show that it satisfies \Cref{claim:syb} as part of an inductive argument (since this property can be used to guarantee safety of the committee up to a certain reconfiguration step $j$, which in turn implies that this property also applies when $\mathcal{C}_j$ performs reconfiguration).
	We will additionally be interested in spam mitigation, since the position registration process is effectively fully permissionless. To be more precise, we ensure the following.
	\begin{proper}[Quantum Spam-Resistance] \label{prop:spam_resistance}
		\begin{enumerate}[i)]
			\item \label{prop:spam_invalid} A message $m$ that is invalid (either not signed by a valid authority or not containing a correct solution to the discrete logarithm puzzle) can be identified by an efficient, local, classical computation.
			\item A message $m$ that contains a correct solution to the discrete logarithm puzzle cannot be efficiently produced without cost (i.e. without access to a quantum computer).
			\item Any honest party with a quantum computer can send valid messages.
		\end{enumerate}
	\end{proper}
	The importance of \Cref{prop:spam_invalid} above is that it allows any participant to detect an invalid discrete logarithm solution locally and prevent such messages (which can be from any source) from being propagated further in the network, since the public ledger used by the protocol is usually implemented using a gossip protocol. This is analogous to invalid claims of a solution to a PoW puzzle being detectable locally, and prevents the overhead of propagating such messages to the rest of the network. 
	
	Safety and liveness of Solida, conditioned on safety of each committee, follow directly from results for the corresponding PBFT protocol. We leave this unchanged in our version:
	
	\begin{theorem}[Safety and Liveness of Solida~\cite{AbrahamMNRS17}] The protocol achieves safety and liveness if each committee has no more than $f < n/3$ Byzantine members.
	\end{theorem}
	
	In order to modify Solida to work with position verification, we only need to modify Theorem 2:
	\begin{theorem}[Safety and Liveness of Solida Reconfiguration~\cite{AbrahamMNRS17}] \label{thm:solida_sec}
		Assuming $f< n/3$ holds for $\mathcal{C}_1$, then $f < n/3$ holds for each subsequent committee except for a probability exponentially small in $n$ if $\rho' (\rho, D, \Delta) < 1/3$\footnote{This theorem as stated in~\cite{AbrahamMNRS17} cannot hold for arbitrarily large times $T$. Rather, this result applies for any specific time slot $t$, and can then be uniformized.}. 
	\end{theorem}
	
	Here $\rho' \geq \rho$ is the effective proportion of external resource held by Byzantine nodes. It is a function of $\rho$, the maximal message delay $\Delta$ and the expected time for the network to solve the PoW puzzle $D$, and increases with $\Delta$.

	Denote by $f_i$ the number of Byzantine committee members on the committee $\mathcal{C}_i$. This is a random variable that depends on the randomness of both position verification and sampling a location from the list of registered locations $\mathcal{E}.\textsf{keys}$. We show the following:
	
	\begin{theorem}[Safety of Committee in \Cref{alg:hcqpv}] \label{thm:safety_committee} Denote by $K$ an upper bound on the number of candidate locations $|\mathcal{E}.\textsf{keys}|$. 
		Assume that for some constant $\varepsilon < 1/3$,
		\begin{enumerate}[i)]
			\item 
			\begin{equation}
				f_1 < (1/3-\varepsilon)n.
			\end{equation}
			\item 
			\begin{equation}
				\rho < 1/3-\varepsilon.
			\end{equation}
			\item CVPV is secure with security parameter $\lambda$, and  
			\begin{equation}
				\lambda>\log K.
			\end{equation}
			\item 
			\begin{equation}
				n > \lambda^c
			\end{equation}
			for some absolute constant $c$ (which can be smaller than $1$).
		\end{enumerate}
		Then there is a constant $c(\varepsilon)$ such that for any $T = \poly(\lambda)$ and sufficiently large $n$, the protocol in ~\Cref{alg:hcqpv} obeys
		\begin{equation}
			\mathbb{P}\left[\forall t\leq T:f_{t}<n/3\right]\geq1-e^{-c(\varepsilon)\lambda}.
		\end{equation}
	\end{theorem}
	\begin{proof}
		
		The assumption on $\lambda$ implies that (for any adversary model chosen and large enough $\lambda$, see \Cref{sec:pv} for details), the success probability of an adversary breaking a single instance of CVPV is at most $e^{-c'\lambda}$ for some absolute constant $c'$. At each reconfiguration step, each committee member runs CVPV with each of the at most $K$ candidates (since the while loop in \Cref{step:cand_elec} of \Cref{alg:hcqpv} terminates after at most $K$ steps). Therefore, by a union bound, the probability that security of CVPV holds for all $1 \leq t \leq T$ is at least 
		\begin{equation}
			1-KnTe^{-\lambda}>1-e^{-c\lambda}
		\end{equation}
		for some $c,c'$ and sufficiently large $\lambda$. The randomness here is internal to CVPV and independent of any other randomness in the protocol. We henceforth condition on the event $\mathcal{E}'$ that CVPV is secure for all $T$ reconfigurations.
		
		When sampling an eligible candidate location from $\mathcal{E}.\textsf{keys}$ (Step \ref{step:sample_cand} of \Cref{alg:hcqpv}), the probability that a Byzantine node is sampled at reconfiguration step $i$ is at most $\rho$, and these events are independent between rounds (since the sampling is independent between rounds, and by our assumption on the resource distribution, honest nodes holding a fraction at least $1 - \rho$ of the quantum computers will also be registered in $\mathcal{E}$). We will argue that as a consequence, as long as at least $2/3$ of the committee members at step $i$ are honest, the probability that $f_{i+1} \geq n/3$ is negligible in $\lambda$. 
		
		Consider the event in which a committee at any reconfiguration step $t < T$ is \textit{unsafe}, meaning $f_t \geq n/3$. We can bound the probability of this event by partitioning the corresponding sample space based on the first time step at which this condition was satisfied. Since these events are all disjoint and their union is simply the event that an unsafe committee was formed at some time, this is a valid partition. This gives
		\begin{equation} \label{eq:sum_unsafe}
			\begin{aligned} & \mathbb{P}\left[\exists t\leq T:f_{t}\ge n/3|\mathcal{E}'\right]\\
				= & \underset{t=1}{\overset{T}{\sum}}\mathbb{P}\left[\left.f_{t}\ge n/3\cap\underset{i=1}{\overset{t-1}{\bigcap}}f_{i}<n/3\right|\mathcal{E}'\right]\\
				= & \underset{t=1}{\overset{T}{\sum}}\mathbb{P}\left[f_{t}\ge n/3\left|\mathcal{E}'\cap\underset{i=1}{\overset{t-1}{\bigcap}}f_{i}<n/3\right.\right]\mathbb{P}\left[\left.\underset{i=1}{\overset{t-1}{\bigcap}}f_{i}<n/3\right|\mathcal{E}'\right]\\
				\leq & \underset{t=1}{\overset{T}{\sum}}\mathbb{P}\left[f_{t}\ge n/3\left|\mathcal{E}'\cap\underset{i=1}{\overset{t-1}{\bigcap}}f_{i}<n/3\right.\right]
			\end{aligned}
		\end{equation}
		
		Consider a single term in the above sum at round $t$. In rounds $t-n, \dots, t-1$, a Byzantine node is added to the committee iff it is sampled from $\mathcal{E}.\textsf{keys}$ (since the election process succeeds deterministically conditioned on $\mathcal{E}'$ and all committees were safe up to $t$). Then, due to independence of sampling between rounds, $f_t$ is a random variable equal to a sum of $n$ independent Bernoulli variables, each with parameter (at most) $\rho$. We then have 
		\begin{equation}
			\begin{aligned} & f_{t} & = & f_{t-1}-\mathbbm{1}\left[\text{Winner at }t-n\text{ was Byzantine}\right]\\
				&  &  & +\mathbbm{1}\left[\text{Winner at }t\text{ is Byzantine}\right].\\
				& \mathbb{E}\left[f_{t}\right] & = & \mathbb{E}\left[f_{t-1}\right]-\mathbb{P}\left[\text{Winner at }t-n\text{ was Byzantine}\right]\\
				&  &  & +\mathbb{P}\left[\text{Winner at }t\text{ is Byzantine}\right]\\
				&  & = & \mathbb{E}\left[f_{t-1}\right]-\rho+\rho\\
				&  & = & \mathbb{E}\left[f_{t-1}\right]\\
				&  & = & \mathbb{E}\left[f_{0}\right]\\
				&  & = & n\rho.
			\end{aligned}
		\end{equation}
		In the above, if $t-n<1$ we interpret ``$\text{Winner at }t-n\text{ was Byzantine}$'' to mean ``the $t$-th member of the Genesis committee is Byzantine''. It then follows from Chernoff's inequality that
		\begin{equation}
			\mathbb{P}\left[f_{t}\geq(1+\delta)n\rho\left|\mathcal{E}'\cap\underset{i=1}{\overset{t-1}{\bigcap}}f_{i}<n/3\right.\right]\leq e^{-\delta n\rho\ln(1+\delta)/2}.
		\end{equation}
		Picking $\delta = 1/(3\rho)-1$ and using $\rho < 1/3-\varepsilon$ gives 
		\begin{equation}
			\begin{aligned} & \mathbb{P}\left[f_{t}\geq n/3\left|\mathcal{E}'\cap\underset{i=1}{\overset{t-1}{\bigcap}}f_{i}<n/3\right.\right] & \leq & e^{n(1-3\rho)\ln(3\rho)/6}\\
				&  & \leq & e^{-c(\varepsilon)n}
			\end{aligned}
		\end{equation}
		for a constant $c(\varepsilon)$ that is independent of $n$.
		Plugging this into \Cref{eq:sum_unsafe} gives
		\begin{equation}
			\begin{aligned} & \mathbb{P}\left[\exists t\leq T:f_{t}\ge n/3|\mathcal{E}'\right] & \leq & Te^{-c(\varepsilon)n}\\
				&  & \leq & Te^{-c(\varepsilon)\lambda^{c_0}}\\
				&  & \leq & e^{-c'(\varepsilon)\lambda^{c_0}}
			\end{aligned}
		\end{equation}
		for appropriately chosen $c'(\varepsilon)$ and large enough $\lambda$, where we used $T=\poly(\lambda)$ and $n > \lambda^{c_0}$. Combining this with the bound on $\mathcal{E}'$ not holding that was shown earlier gives
		\begin{equation}
			\begin{aligned} & \mathbb{P}\left[\exists t\leq T:f_{t}\ge n/3\right] & \leq & \mathbb{P}\left[\exists t\leq T:f_{t}\ge n/3|\mathcal{E}'\right]+\mathbb{P}\left[\mathcal{E}'^{c}\right]\\
				&  & \leq & e^{-c(\varepsilon)\lambda^{c_0}}+e^{-c'\lambda}\\
				&  & \leq & e^{-c''(\varepsilon)\lambda^{c_0}}
			\end{aligned}
		\end{equation}
		for appropriate constants and sufficiently large $\lambda$.
	\end{proof}
	
	Note that in this result, the effective amount of resource held by honest players does not depend on $\Delta$. This is in contrast to PoW-based protocols~\cite{Sompolinsky2015-bk} and others based on longest-chain consensus~\cite{Dembo2020-vd}. The security of the protocol will still depend on computation time, and arguably the spatial resolution of the protocol is somewhat analogous to the message delay and will degrade security, but propagation delays in themselves do not impact security (and are in fact essential for security of position verification).

	\begin{theorem} \label{thm:liveness} Assume the registration time $\tau_{\text{register}}$ is chosen so that an honest party can complete the DLOG puzzle with security parameter $\kappa$. 
		\Cref{alg:hcqpv} satisfies liveness, with spam resistance (\Cref{prop:spam_resistance}) in the Random Oracle model.
	\end{theorem}
	
	\begin{proof}
		Any transaction submitted during $\textsf{SteadyState}$ or $\textsf{ViewChange}$ is guaranteed to be processed by the liveness guarantee of the PBFT protocol used by Solida%
		. It remains to check liveness of registration (\Cref{alg:pq_reg_rom}). This is satisfied trivially in the vanilla version (\Cref{alg:pq_reg_rom_spam}), albeit adversarial nodes may send multiple registration messages without possessing quantum computers at the corresponding locations. This will lead to overhead for the committee when attempting to verify these positions. 
		We next consider the spam-resistant version. We have the following:
		
		\textit{Completeness:} Since $H(\textsf{pk} \mathbin\Vert \textsf{pos} \mathbin\Vert r)$ is a $\poly(\lambda)$-bit integer, the complexity of solving the discrete logarithm using Shor's algorithm is $\poly(\lambda)$~\cite{Shor1994-hb}.  By Assumption \ref{ass:bounded_qcomp}, any honest party in possession of a quantum computer may publish a registration message. The adversary cannot prevent any honest party from submitting a registration message, and it will be processed by every other honest node.
		
		\textit{Soundness:} %
		By the hardness assumption of the discrete logarithm problem, there is no classical algorithm for solving it in time $\poly(\lambda)$. Thus an adversary without a quantum computer cannot costlessly publish valid registration messages. 
		
		The size of the preliminary candidate set $\mathcal{E}.\textsf{keys}$, denoted by $K$, is also finite due to Assumption \ref{ass:bounded_qcomp} and the limited time $\tau_{\text{register}}$, and the runtime of each reconfiguration is $O(K)$. Every other stage in the protocol terminates in finite time, and every message arrives with delay at most $\Delta$. %
	\end{proof}
	
	Combining these results, we have the following:
	\begin{corollary}
		Assume
		\begin{enumerate}[i)]
			\item Sub-exponential quantum hardness of LWE. \label{ass:lwe}
			\item The adversary can pre-share at most $O(\lambda^k)$ Bell pairs for some known $k$. \label{ass:entanglement}
			\item 
			\begin{equation}
				f_1 < (1/3-\varepsilon)n.
			\end{equation}
			\item 
			\begin{equation}
				\rho < 1/3-\varepsilon.
			\end{equation}
			\item 
			\begin{equation}
				n > \lambda^c
			\end{equation}
			for some $0<c<1$.
		\end{enumerate}
		Then, for any $T=\poly(\lambda)$, \Cref{alg:hcqpv} satisfies consistency and liveness (\Cref{prop:consistency_liveness}), and guarantees Sybil resistance as defined in \Cref{claim:syb}. 
		
		Additionally, if %
		we assume 
		\begin{enumerate}[i)]
			\addtocounter{enumi}{5}
			\item The Random Oracle assumption.
			\item $\lambda > \log(T R_A(\lambda)/R_H(\lambda))$, \label{ass:lambda_ineq}
		\end{enumerate}    
		then \Cref{alg:hcqpv} also satisfies spam-resistance (\Cref{prop:spam_resistance}).
	\end{corollary}
	\begin{proof}
		By assumptions \ref{ass:lwe} and \ref{ass:entanglement}, we can construct secure CVPV with security parameter $\lambda$.
		
		Set $\kappa = \lambda$. Liveness of the spam-resistant version is guaranteed by picking 
		\begin{equation}
			\tau_\text{register} = 1/R_H(\lambda),
		\end{equation}
		so that honest nodes can register their positions. This implies the adversary can make at most $T R_A(\lambda)/R_H(\lambda)$ registrations in this time period. Hence $K=R_A(\lambda)/R_H(\lambda)$ and Assumption \ref{ass:lambda_ineq} guarantees $\lambda > \log(K)$. Note that 
		\begin{equation}
			\lambda > \log(T R_A(\lambda)/R_H(\lambda))
		\end{equation}
		can be satisfied for sufficiently large $\lambda$ since the problem is efficiently solvable using a quantum computer and hence $R_A,R_H,T$ are polynomials.
		
		Consistency of $\textsf{SteadyState}$ and $\textsf{ViewChange}$ follows from the safety of the committee at all times, which is guaranteed by \Cref{thm:safety_committee}. All steps in the registration process are publicly verifiable conditioned on the initial committee $\mathcal{C}_1$, and thus consistency of the registration process is also assured. Liveness follows from \Cref{thm:liveness}.
	\end{proof}
	
	We reiterate that the bounded pre-shared entanglement assumption can be replaced by any polynomial entanglement, at the price of moving to the QROM. Note that because of the mild (logarithmic) dependence on $R_A,R_H$, it is sufficient to use very loose bounds on these functions that need not be highly sensitive to implementation details. 
	
	\subsection{Notable properties} \label{sec:notable_properties}
	
	Having shown that the protocol satisfies safety and liveness, we now analyze various aspects and highlight advantages relative to existing classical protocols.

	\subsubsection{Non-rivalry} \label{property:nonrival}

    Budish identifies two scale economies enjoyed by traditional trust, the kind grounded in rule of law, and absent from what he calls Nakamoto trust~\cite{Budish2025-fy}. A criminal contemplating a bank robbery must first overcome the guards at the branch, then the police reinforcements the bank will call, and finally the punishment if caught. Nakamoto trust ``is just to have a very large number of security guards at the bank'', and moreover the same number ``outside the local bank branch as outside Fort Knox or the Federal Reserve''. The first scale economy is the police: they ``do not have to be present at the particular bank to provide security---they can provide security to many locations at once as long as they are not all attacked simultaneously''. The second is the courts, in the spirit of Becker's economics of crime~\cite{Becker68}: a fixed investment in court capacity deters any amount of crime, because a credible threat of a large penalty costs nothing so long as no one tests it. Slashing in Proof-of-Stake mimics the second, ``trying to mimic the traditional trust combination of collateral plus rule of law'' in Budish's words, within the limits established in~\cite{Budish2024-sm}. It is the first that concerns us here: QPoP recovers it, and is to our knowledge the first permissionless Sybil-resistance mechanism to do so without relying on a trusted party.

    In the language of economics, guards are a \emph{rival} resource and police are \emph{non-rival}: a resource is rival if one party's use of it leaves less for everyone else, as with a guard, a sandwich, or a unit of hash rate, and non-rival if it can serve any number of parties at once without being used up, as with a police force, a lighthouse, or a theorem. Hash rate is a guard: a unit of it devoted to one chain protects no other, so every PoW chain must hire its own guards, in full, at every instant. Stake is a guard that cannot even change branches, since it is ``literally units of the cryptocurrency''~\cite{Budish2025-fy}. 

    A quantum computer in QPoP is police. Position verification occupies it only briefly and commits nothing to a particular protocol in between, so the same device can be registered with any number of protocols, each enjoying its full benefit rather than a fraction. The trust support is thus non-rival up to the duty cycle of the device, the exact analogue of Budish's proviso that the locations not all be attacked at once, and its capital cost is amortized across every chain it secures and, being a stock rather than a flow, across time as well.\footnote{Two clarifications. What is non-rival is the security service supplied to distinct protocols, not the hardware, which is rival in the ordinary sense; the same is true of police officers. And non-rivalry is only half the definition of a public good: QPoP trust support remains fully \emph{excludable}, since the owner chooses which protocols to register with and is paid by each, which is what preserves the incentive to build the hardware.}

    Rivalry is not merely expensive in PoW; it is the root of several of its instabilities, since a miner must choose \emph{one} chain, or none. The option of none gives rise to the mining gap of \Cref{property:security}. The choice of which couples a chain's security to its price: Bissias et al.~\cite{Bissias2018-qj} find that the Bitcoin Cash inter-block times predicted by their model of economic hash-rate allocation correlate at $0.849$ with the realized ones, an effect they liken to ``credit card transactions being processed more slowly whenever the stock price of Visa Inc.\ drops''. And since miners share a shutdown point, aggregate hash rate is highly elastic in the reward near it; Noda et al.~\cite{NodaOH26} show that Bitcoin's difficulty adjustment is stable only if this elasticity is below one, so that a sharp price drop can leave the chain producing blocks far below its target rate. The correction is itself slow to arrive, since the difficulty is recalibrated only every $2016$ blocks and few blocks are found while the chain is stalled~\cite{Bissias2018-qj}; in the extreme this is the \emph{chain death spiral}~\cite{Bissias2020}. Bitcoin Cash came close to it under its original $2016$-block adjustment, at times taking about four hours to produce a block~\cite{Kwon2019}. None of these arise when the trust support is non-rival, because there is nothing to choose.

    Non-rivalry has been sought before, but always by coupling chains or by trusting a party. \emph{Merged mining}~\cite[Chapter 10]{NarayananBFMG16} reuses one (partial) PoW solution across a parent chain and its children, which raises the children's difficulty~\cite{JudmayerZSVW17} but requires each child to adopt the parent's PoW algorithm and verify its proofs~\cite[Chapter 10]{NarayananBFMG16}, while the parent ``cannot easily prevent being merge-mined''~\cite{JudmayerZSVW17}; since a merge-mined block costs the parent miner nothing beyond running a node for the child, attacking a child is nearly free, which is how the Eligius pool stalled and ultimately killed off CoiledCoin in 2012 without its own members noticing~\cite{JudmayerZSVW17, NarayananBFMG16}, and a single pool has held a majority of the hash rate of several merge-mined chains for prolonged periods~\cite{JudmayerZSVW17}. 
	\emph{Restaking} lets the same bonded stake back many services and inherits the same pathology, as validators ``can coordinate to misbehave across multiple services simultaneously, extracting digital assets while forfeiting their stake only once''~\cite{BarZurE25}.
	Proof-of-personhood~\cite{BorgeKJGGF17} and elapsed-time protocols on trusted hardware~\cite{Bowman2021-mr} are non-rival but obtain it from a trusted set of identity issuers or from a hardware manufacturer, respectively, which we regard as a departure from permissionlessness. In QPoP each protocol samples positions and runs CVPV entirely on its own: the proofs are not shared, only the device is, and non-rivalry follows from the uncloneability of quantum states rather than from any agreement between chains or any trusted party.

	\subsubsection{Geographic decentralization} \label{property:geographic}
	
	Decentralization is hard to sustain in practice, since most forces acting on a deployed protocol favor concentration. Nearly every Sybil resistance resource exhibits economies of scale: bulk purchasing favors ever larger Proof-of-Work facilities, reward variance drives participants into pools~\cite{Fanti2019-nj}, and Proof-of-Stake rewards compound with existing holdings~\cite{Fanti2019-nj}. Bitcoin and Ethereum are consequently centralized along several axes, geography among them~\cite{Gencer2018-mo, Sai2021-uc}.
	
	QPoP has a feature that, to the best of our knowledge, no other Sybil resistance mechanism possesses: its resource cannot be agglomerated, so the protocol \emph{generates} geographic decentralization rather than merely tolerating it. By \Cref{claim:syb}, a participant's power is the number of \emph{distinct cells} of $\mathcal{P}$ in which they operate an online quantum computer, and a second computer in an occupied cell adds nothing, since a device passes CVPV only from where it actually sits. Hashrate and stake are by contrast perfectly agglomerable: a gigawatt of mining capacity fits inside a single warehouse, and any amount of stake inside a single key. Even a participant holding many quantum computers must therefore spread them across space, at a minimum separation of $\Gamma$.
	
	The designer can push diversification further by making the block reward increase with a position's distance from the other registered positions. Such distances are public and unforgeable, since positions in QPoP are cryptographically attested rather than inferred from spoofable signals such as IP geolocation or network latency~\cite{Roeschlin2026-xw}. Rewarding a position in proportion to the area of its Voronoi cell, for instance, equalizes reward density over space rather than over participants, and is self-policing in that registering close to an existing device shrinks the reward of both. It is not without weaknesses: an attacker could mount a \emph{QPoP siege}, placing three computers around a victim to bound its Voronoi cell. We leave such directions to future work.
	
	The absence of geographic decentralization has repeatedly proved consequential. China long hosted the largest share of Bitcoin's hashrate, $75.5\%$ in September $2019$ and still $46\%$ in April $2021$, and one sequence of policy decisions drove its measured share to zero within weeks, though the measured share later rebounded above $20\%$, which is attributed to underground mining~\cite{CCAF-map}. Miners relocated largely to the United States and, outside it, to Kazakhstan, where the internet shutdown during the unrest of January $2022$ removed some $14\%$ of global hashrate almost immediately~\cite{Reuters22}; Kosovo and Iran have likewise suspended miners during energy shortfalls~\cite{Tully22}.

	Proof-of-Stake is exposed in the same way. After the U.S. Office of Foreign Assets Control sanctioned the Tornado Cash addresses in 2022, $46\%$ of Ethereum blocks were built by parties intending to comply with the sanctions~\cite{WahrstatterEYZQ24}, and one hosting provider's ban on cryptocurrency workloads knocked over a thousand Solana validators offline at once, roughly a fifth of the network's stake having been hosted there~\cite{Lutz22}. In each of these cases the problem was concentration rather than a flaw in the protocol itself. Geographic decentralization may mitigate such issues, since no single jurisdiction or provider would then control more than its share of eligible positions.
	
    \subsubsection{Capital intensiveness} \label{property:capex}
	
    In a PoW cryptocurrency, mining is subject to free entry, and consequently the block rewards are dissipated almost entirely on the expenses of mining. Following~\cite{Budish2025-fy}, denote by $c$ the cost per unit of time of supplying one unit of the resource on which Sybil resistance is based (one unit of hash rate, in the case of PoW), by $R^\ast$ the total amount of that resource supplied in equilibrium, and by $d$ the reward per block. Since participants may enter and exit at will, in equilibrium they earn zero economic profits, which gives
    \begin{equation} \label{eq:equilibrium}
        d = R^\ast c.
    \end{equation}
    This is, as Budish notes, ``the standard characterization of a rent-seeking tournament'', in which the prize ``is dissipated by expenditures aimed at winning the prize''~\cite{Budish2025-fy}, and Prat and Walter~\cite{PratW21} provide empirical evidence that \Cref{eq:equilibrium} describes the actual behavior of the Bitcoin mining market. Should participants in fact earn positive profits, the $R^\ast$ of \Cref{eq:equilibrium} is an upper bound~\cite{Budish2025-fy}, and the discussion below is unaffected.

    The cost $c$ decomposes as $c = rC + \eta$, where $C$ is the capital cost of one unit of the resource, $r$ is its rental cost per unit of time, and $\eta$ is the operating cost per unit of time.
	We write $\theta = rC/(rC+\eta)$ for what Budish calls the \emph{capital share of trust support}, which he calibrates at $0.4$ for Bitcoin~\cite{Budish2025-fy}, so that neither component dominates: by \Cref{eq:equilibrium}, PoW block rewards are split between acquiring hardware and running it. QPoP, in contrast, calls for quantum hardware only rarely. The sole steps of the protocol that require a quantum computer are \Cref{step:cvpv_proof} of \Cref{alg:pv_comm} and \Cref{step:dlog_poq} of \Cref{alg:pq_reg_rom}, and each is carried out by a party attesting to a position of its own.
	A reconfiguration round therefore puts exactly one quantum computer to work, that of the single candidate under verification, while the committee members serve as classical verifiers and every remaining step of the protocol, including all of those executed by the committee, is classical as well. Crucially, this figure of one is unaffected by how large the committee grows, and equally by how many participants the protocol has. Consequently $\eta$ is negligible compared to $rC$, and the cost of participation is dominated by the capital expenditure of building the hardware. The shift is thus from a mixture of capital and operating expenses, in which neither dominates, to a cost structure that is almost purely capital: $\theta$ moves from roughly one half in PoW to nearly $1$ in QPoP. By \Cref{eq:equilibrium}, the block rewards of QPoP are correspondingly spent almost entirely on acquiring quantum hardware. %
        
	\subsubsection{Energy efficiency} \label{sec:energy_efficiency}

	The shift in cost structure established in \Cref{property:capex} implies a drastic reduction in energy consumption. PoW protocols, and Bitcoin in particular, are notorious for their energy consumption~\cite{Zhang2020-jx, Sedlmeir2020-qi}, and the electricity they consume carries a substantial carbon footprint: Bitcoin's annual emissions were estimated at $22$--$23$ MtCO$_2$ as of November 2018~\cite{StollKG19}, and at roughly $65$ MtCO$_2$, about $0.19\%$ of global emissions, as of August 2021~\cite{DeVriesGKS22}, the increase driven in part by the migration of mining away from China towards more carbon-intensive grids. A recent global assessment finds substantial water and land footprints as well~\cite{ChamanaraGM23}. In the terminology of Tsabary and Eyal~\cite{TsabaryE18}, who parameterize a mining market by the split of its costs between operating and capital expenses, Bitcoin sits in the \emph{med\_op} regime, where the two are of comparable magnitude: they calibrate an Antminer S9 at roughly \$1{,}000 of hardware against some \$876 of electricity per year, and conclude that ``the system falls in the area of med\_op''. QPoP instead sits in their \emph{low\_op} regime, in which operating expenses vanish and all costs are capital. This is in contrast to Proof-of-Work protocols and in particular the Bitcoin protocol~\cite{Nakamoto2008-bc}, in which the constant mining required by all participants is a well-known drawback. The magnitude of the reduction can be gauged from Ethereum's transition to PoS, which cut its energy consumption by $99.988\%$~\cite{CCRI22}; we expect a transition from PoW to QPoP to have a similar impact.

    While the above claims relate to asymptotic energy consumption, we are in the process of conducting detailed resource estimates for a variant of the protocol that uses efficient proofs of quantumness based on ring-LWE~\cite{Brakerski2020-yo}. Preliminary results indicate that this leads to energy consumption at parity with a pure PoS protocol such as Ethereum, while retaining the advantages of using an external resource.%
	
	\subsubsection{Security} \label{property:security}

	In hybrid protocols based on PoW like Solida, the time between reconfigurations is also set by the difficulty of the PoW puzzle. Since we use a different external resource, there is no coupling between the reconfiguration frequency and the properties of the network or security. At the limit of the frequency going to zero, the protocol approaches a fully permissioned setting with a fixed committee. As it is increased, new participants will be able to participate with less delay, albeit at the cost of increased resource requirements. %

    Being CapEx dominated, as established in \Cref{property:capex}, carries several advantages for security. In OpEx-dominated protocols, a reduction in asset values can lead to participants exiting the protocol by simply stopping mining, which compromises security at the precise times when it is most needed, namely in periods of high volatility; Bissias et al.~\cite{Bissias2018-qj} document this empirically, as discussed in \Cref{property:nonrival}. In CapEx-dominated protocols such as QPoP by contrast, participants are ``locked in'' by the cost of initial participation, and security can therefore not be compromised as easily when prices are volatile.

    We note that PoS is also CapEx-dominated. Exiting a PoS position has some friction, since it requires finding a buyer for one's stake and the sale itself incurs some loss. In the case of QPoP, the barrier to exit is substantially larger, since finding a buyer for a quantum computer will generally be less trivial, and deprecation will likely lead to a significant sunk cost compared to PoS. For the health of the protocol as a whole this is a beneficial feature, since it improves security by making it more difficult for participants to exit their positions. The downside of CapEx-dominated protocols is that the barrier of entry for individual parties is higher, though arguably it is already somewhat high for modern PoW protocols.

    A perhaps surprising further advantage of being CapEx dominated is the absence of a \emph{mining gap}. Carlsten et al.~\cite{CarlstenKWN16} postulated that if block rewards originate only from transaction fees, then miners would avoid mining immediately after a block is found, when no fees have yet accrued, leaving the network's mining power idle for a period of time. Tsabary and Eyal~\cite{TsabaryE18} studied this phenomenon as a game, and found that such gaps form well before fees are the only incentive, that they can decrease mining power utilization by up to $90\%$ in the most extreme scenario they consider, and that they incentivize miners to form coalitions, thereby harming decentralization. The mechanism behind the gap is that a miner who idles saves operating expenses at the cost of a reduced probability of finding the next block, so the driving force is the OpEx component in particular: as \cite{TsabaryE18} observe, the model of Carlsten et al.\ ``assumes only operational expenses and no subsidy''. Accordingly, in the regime of~\cite{TsabaryE18} in which all expenses are capital, every miner starts mining immediately at equilibrium, for every composition of the block reward, and mining power utilization is unaffected; in fact, no base reward at all is then needed in order to prevent a gap. Since the marginal cost of participating in QPoP is negligible, we expect no analogue of the mining gap to arise. We stress that~\cite{TsabaryE18} analyze PoW protocols and explicitly caution that their results do not automatically carry over to other settings, so the above should be read as a consequence of their model rather than as a formal claim about QPoP.

    Under an additional assumption that the market for quantum computers is illiquid, QPoP also features improved security in an economic sense, as studied by~\cite{Budish2024-sm, Budish2025-fy}. These works attempt to model security of permissionless protocols not just as functions of the number of Byzantine players or the amount of resource they control, but directly in terms of the economic costs of an attack. We discuss their results at a high level.
    
    Recall from \Cref{property:capex} the equilibrium condition \Cref{eq:equilibrium}, namely $d = R^\ast c$, where $R^\ast$ is the amount of resource supplied in equilibrium, which is outside the direct control of the protocol in the permissionless setting, and recall the decomposition $c = rC + \eta$ of the cost per unit of hardware per time step into its capital and operating components.

	If we additionally assume that the number of units needed to attack the protocol is $AR^\ast$ for some $A$, and crucially that the asset is illiquid, meaning that it needs to be purchased specifically for such an attack and cannot be rented, we have
    \begin{equation}
        \begin{aligned} & \text{AttackCost}_{\text{illiquid}} & = & AR^{\ast}C\\
 &  & = & \frac{AdC}{c}\\
 &  & = & \frac{AdC}{rC+\eta}\\
 &  & = & \frac{Ad}{r+\eta/C},
\end{aligned}
    \end{equation}
    where we used \Cref{eq:equilibrium} in the second line. Since for QPoP $\eta/C$ is negligible (because it is CapEx-dominated) while for PoW $\eta/C$ will generally dominate $r$, assuming equal block rewards and interest rates we have 
    \begin{equation}
\text{AttackCost}_{\text{illiquid}}(\text{PoW})=\frac{Ad}{r+\eta/C}<\frac{Ad}{r}\approx\text{AttackCost}_{\text{illiquid}}(\text{QPoP}).
    \end{equation}
    This implies that in the illiquid regime QPoP will have improved security. This will also be true if the illiquidity is only approximate, assuming other factors are equal. The ratio of the attack cost on QPoP vs PoW is exactly the inverse of the capital share of trust support $\theta$, which Budish calibrates at $0.4$ for Bitcoin, which implies that in the illiquid regime
    \begin{equation}
        \frac{\text{AttackCost}_{\text{illiquid}}(\text{QPoP})}{\text{AttackCost}_{\text{illiquid}}(\text{PoW})} \gtrsim 2.5.
    \end{equation}
    
    Note that stake in PoS is essentially liquid by definition, and so in that case an attack that takes time $t_{\text{attack}}$ costs $\text{AttackCost}_{\text{liquid}} = AR^\ast ct_{\text{attack}}=Adt_{\text{attack}}$ (which is the regime studied in~\cite{Budish2024-sm}). Consequently, the factors of $c$ cancel out and so even though PoS is also CapEx dominated this does not affect the cost of an attack. Thus stake can only be made illiquid by clever protocol design, for example by locking up stake and enforcing cooldown periods to enable slashing after an attack. This introduces additional complexity and possible attack vectors. The authors of \cite{Budish2024-sm} show that slashing leads to costly attacks in their model. It is worth noting that it might be possible to combine QPoP with a slashing mechanism acting on on-chain resources that is decoupled from the position-based Sybil resistance mechanism, which might have useful properties, though we leave this for future work.
    
	\subsubsection{Deterministic finalization} \label{property:finalization}

	Transactions in QPoP are finalized by a deterministic BFT protocol run by the committee, and the protocol therefore inherits the fast-confirmation advantage of hybrid committee protocols~\cite{PassS17,KokorisKogiasJ16,AbrahamMNRS17}. Conditional on the committee satisfying the honest super-majority Byzantine fault threshold, finalized transactions cannot be rolled back~\cite{KokorisKogiasJ16,AbrahamMNRS17}. This is in contrast to Bitcoin, where confirmation is probabilistic and rollbacks always remain possible~\cite{Nakamoto2008-bc}: a transaction is only ever confirmed to within some error probability, which decays with the number of blocks mined on top of it, and the customary practice of waiting for six confirmations is a convention rather than a guarantee. The distinction is a consequence of the Sybil resistance mechanism rather than of the consensus protocol layered on top of it. Since the winner of a PoW puzzle is determined by a stochastic process, an adversary controlling a minority of the hash rate retains a non-zero probability of producing a longer chain, whereas the committee of QPoP is fixed and known at each point in time, so that agreement among its members is final.

	\subsubsection{No obsolescence} \label{property:obsolescence}

	Mining hardware becomes obsolete rather quickly. An ASIC is not retired because it breaks, but because newer models perform the same computation using less electricity, and the useful life assumed by miners is correspondingly short: Budish~\cite{Budish2025-fy} takes an annual discount rate of $50\%$ for ASICs in his calibration, and listed miners have recently shortened the depreciation schedules of their machines to as little as three years. The mechanism of retirement is purely economic. A machine remains in service exactly as long as the rewards it earns exceed the electricity it consumes; once the price of the coin falls, or the difficulty rises, or a more efficient model captures a large enough share of the hash rate, the older machine is switched off and is worth little to anyone else, since it can do nothing but compute the same hash function.

	In QPoP this mechanism disappears, because there is no operating expense for it to act on. A quantum computer that has been registered at a position contributes to Sybil resistance whether or not more capable machines have since been built, and its owner has no reason to withdraw it, since keeping it registered costs essentially nothing. The improvement is precisely the one identified in \Cref{property:capex}: obsolescence in PoW operates through the operating cost $\eta$, and QPoP has $\eta \approx 0$, so that its capital share of trust support is $\theta \approx 1$. Recall from \Cref{eq:equilibrium} that the resource supplied in equilibrium satisfies $R^\ast = d/c$, so with $c = rC + \eta$ the amount of hardware securing the protocol is $d/(rC)$ rather than $d/(rC + \eta)$: every unit of the block reward buys strictly more security when none of it is burned as electricity. Put differently, a dollar spent on quantum hardware keeps securing the protocol indefinitely, whereas a dollar spent on electricity secures Bitcoin only for the instant in which it is consumed. This is the sense in which Nakamoto trust is, in Budish's phrase, ``memoryless''~\cite{Budish2025-fy}: it ``is as secure at a moment in time as the amount of trust support at that moment'', so that its security must be repurchased continually, while the capital stock of QPoP accumulates.

	\subsubsection{Equitability and censorship resistance} \label{property:equitability}
	
	Since our protocol relies on resources both internal to the protocol and external to it (committee membership and quantum computers, respectively), it circumvents some known issues with protocols that rely solely on internal resources, such as pure Proof-of-Stake protocols. One issue that has been formalized is the lack of equitability of Proof-of-Stake~\cite{Fanti2019-nj}. Equitability is a measure of the variance of the resource in question after a protocol has run for a fixed amount of time. This work shows that when the resource used for Sybil resistance is also the reward (which is stake for Proof-of-Stake), common block reward functions such as constant or decreasing rewards lead to highly inequitable distributions (in some cases, at large times stake concentrates around $0$ and $1$). This is a formal manifestation of the fact that stake earned as block rewards increases a participant's probability of getting selected as the leader in a future round and thus to accumulate more stake. The authors also construct theoretically optimal reward functions, but show that even these are less equitable than those achievable by PoW protocols. The authors of~\cite{Fanti2019-nj} present their results as a fundamental trade-off between the energy efficiency of PoS and the improved equitability of PoW, yet we show that when quantum resources are introduced one could obtain the best of both worlds. Model-checking studies of Ethereum's Gasper protocol report the same dynamics~\cite{Bistarelli2024-oj}.
	
	While PoW protocols also suffer from concentration of wealth, in both PoW and QPoP equitability is at least structurally possible, unlike pure PoS where it is formally limited due to the results of~\cite{Fanti2019-nj}. In practice, the degree of equitability or concentration will depend on complex dynamics external to the protocol itself, which determine the degree of coupling between the internal resource (committee membership and block rewards) and the external resource (computational power, either classical or quantum) since in practice the two are coupled, albeit with sources of friction that do not exist for pure PoS protocols and might retard wealth concentration.

	A related form of wealth concentration in pure PoS protocols has been formalized by Lewis-Pye and Roughgarden~\cite{Lewis-Pye2023-zp}. The example of their Section 11.1 shows that in the partially synchronous setting, stake can become concentrated in the hands of a single participant solely as a consequence of message delays by an adversary, even for a set of transactions that would leave the stake distribution unchanged in the absence of delays; their Theorem 11.1 generalizes this into an impossibility result covering every protocol whose on-chain resource is \emph{reactive} in their sense. Since our protocol requires synchrony, we do not claim a formal advantage as a consequence of this theorem, but instead note it as another example of the fragility of the partially synchronous setting.  

	\subsubsection{Permissionlessness} \label{property:permissionless}
	
	We note that the process of building a quantum computer and registering its position is fully permissionless, in the sense that this can be done independently of any other participant and the committee in particular. As long as the committee is safe, committee members also have no control over which position gets selected as the next committee member. As discussed earlier, this circumvents the wealth concentration issues faced by protocols based on Proof-of-Stake.
	The protocol we present also naturally handles participant inactivity. A participant can run CVPV only when their quantum computer at the sampled location is online. If this is not the case, the protocol proceeds to sampling a different location.

	\section{Resource requirements} \label{sec:resource_requirements}
	
	While our protocol is infeasible to implement on today's hardware and our analysis is asymptotic, one can get a sense of the scale of runtimes and achievable spatial resolution from related resource estimates. 
	
	\subsection{Discrete Logarithm}
	
	Resource estimates from several years ago indicate that computing the discrete logarithm in our setting (i.e., the quadratic residues subgroup of $\mathbb Z_p^*$ for a safe-prime $p$) with 2048 bit requires $7$ hours on a noisy quantum computer with 26 million physical qubits ~\cite[Table 5]{Gidney2021-je}.
	We note that the security parameters required by a protocol like ours will most likely be significantly smaller than $2048$ bits. Various optimizations have reduced the space-time overhead of factoring by a factor of $20$~\cite{Gidney2025-dp}, and while some of these optimizations are tailored to factoring, others (such as T-state cultivation~\cite{Gidney2024-ab}) reduce the cost of generic resource states and would likely imply improved estimates for the discrete logarithm problem as well. Very recent estimates that incorporate these optimizations show that 256-bit ECDLP can be solved in minutes using less than $5*10^5$ physical superconducting qubits~\cite{Babbush2026-ue}.
	
	\subsection{Proofs of Quantumness}
	
	Position verification based on classical communication requires the prover to run proofs of quantumness. There are constructions of highly efficient proofs of quantumness that are not quantum-secure and hence unsuited for our purposes~\cite{Kahanamoku-Meyer2022-ta}. However, quantum-secure, efficient constructions in the ROM have also been developed~\cite{BrakerskiKVV20}, as well as ones in the SM requiring only constant-depth quantum circuits (requiring unbounded fanout gates)~\cite{HiraharaG21}. The proposal of Ref.~\cite{BrakerskiKVV20} is estimated to require $8\log^2(\lambda)\lambda$ qubits to provide a proof of quantumness with security parameter $\lambda$. Indeed, these proposals were used as the basis for demonstrations on ion-trap quantum computers (albeit not in a classically-hard regime)~\cite{Zhu2021-co, PhysRevA.109.012610}, showing that proofs of quantumness are within reach of existing noisy devices. Concrete estimates in these works~\cite{PhysRevA.109.012610} claim that one needs $\approx 10^3$ qubits and circuits of depth $\approx 10^5$ to construct proofs of quantumness secure against classical adversaries. We are currently conducting detailed resource estimates for the efficient ring-LWE based proof of quantumness of~\cite{BrakerskiKVV20}. Using these in place of the original LWE-based ones of~\cite{Liu2022-ci} requires non-trivial technical work, but preliminary results indicate that the energy required by such a scheme would be comparable to a pure PoS protocol like Ethereum. %
	
	\subsection{Position verification}
	
	Quantum position verification has recently been experimentally demonstrated in one spatial dimension, achieving a resolution of about $50\text{m}$~\cite{Kavuri2026-ba}. However, the demonstration did not use classically verifiable position verification, and instead required quantum communication (as in the original formulation~\cite{Buhrman2014-qx}) which is difficult to perform over long distances and large bandwidths. It does serve as an indication that at least in principle, reasonable spatial resolutions can be achieved with this method. 
	
	CVPV will have greatly reduced resource requirements for the verifier. In one spatial dimension, the verifier requires two classical sources and receivers, suitably positioned, that can communicate using optical wireless communication with the prover. Existing classical telecommunications hardware could serve this purpose. The more stringent requirements are those of the prover, who is required to run a proof of quantumness. This is suitable for our protocol since we envisage this as the scarce economic resource conferring Sybil resistance.
    
    In \Cref{app:timing_resolution}, we estimate the spatial resolution that can be achieved by the protocol, and provide more details regarding the resources needed for verification.

	\section{Discussion} \label{sec:discussion}
	
	We have shown that the use of quantum resources in solving distributed consensus has the potential to substantially improve upon some of the shortcomings of classical protocols, in particular those related to energy consumption, censorship, and the assumptions needed for guaranteeing security. While perhaps not immediately implementable, our construction suggests that the uniquely fragile nature of quantum information, which has no classical analog, can be useful in building the economic scarcity that is indispensable for constructing Sybil-resistant decentralized protocols. 
	
	Our protocol serves as a proof of principle that incorporating tools and capabilities from quantum information theory with modern classical approaches to the consensus problem can lead to novel protocols with improved performance. Our work leads to several open questions and possible avenues for improvement:
	\begin{itemize}

		\item Can quantum Sybil resistance (whether based on position verification or otherwise) enable fully permissionless consensus? The best classical attempts require fine-grained complexity results~\cite{BallGHKP24}. The non-transferability of position verification might make such a construction difficult to base on position verification, and indeed this is the reason for including a committee in our protocol, and the corresponding on-chain resource of committee membership. 
		\item Can the position of a participant remain hidden until being elected to the committee? This is reminiscent of recent results on zero-knowledge position verification~\cite{Girish2026-so}, in which a prover can prove some zero-knowledge statement about their position without revealing it. The issue with using such results is that there appears to be a fundamental assumption of an honest verifier in this setting. The verifier effectively runs position verification with all positions at all times, since if they only ran it with a single prover position the prover would inevitably fail and this position could be ruled out. Since we must accommodate Byzantine committee members in our protocol, there is no clear way to work around this. 

        \item One of the main concerns with our protocol is the resource requirements (both in terms of memory and circuit depth) of LWE-based proofs of quantumness, given the time-sensitivity of position verification. On the other hand, since protocol participants are expected to possess quantum computers, it is natural to ask whether a proof of quantumness that can only be verified by a quantum verifier would be sufficient for our purposes. There have been recent experimental demonstrations of problems believed to be classically hard yet efficiently quantumly verifiable, such as the OTOC$^2$ problem~\cite{Google-Quantum-AI-and-Collaborators2025-nq, King2025-lq}. The main difficulty in using OTOC$^2$ to construct position verification is that this is a sampling task and does not enable any mechanism for binding commitment, as in LWE-based proofs of quantumness, which prevent quantum adversaries from spoofing a position by copying classical messages and running identical protocols. This would require inherent nondeterminism in the problem. Random circuit sampling is non-deterministic but not efficiently verifiable. Thus a natural question is whether there exists a problem implementable on near-term hardware that is both non-deterministic and is believed to exhibit quantum advantage? Such a problem could be the basis of proofs of quantumness that one could use with our protocol. More concretely, we also have preliminary results indicating that the resources required could be reduced considerably by using more efficient proofs of quantumness~\cite{Brakerski2020-yo}, yet leave this analysis to future work.
		
		It is worth noting that position verification is possible using quantum communication, which significantly reduces the computational overhead of the prover~\cite{Unruh2014-eb, Kavuri2026-ba}. This is somewhat at odds with the goal of using the quantum computers serving as provers as a costly resource to ensure Sybil resistance, since it instead shifts the burden onto the verifiers, but suggests there could be ways to make such schemes more efficient if needed.
        
		\item Existing CVPV results apply to a single spatial dimension. It would be interesting to extend these to more realistic 2 or 3-dimensional settings, and account for finite computation time.
		
		\item Can the random beacon assumption be removed or weakened? Quantum devices can produce certified randomness (or min-entropy)~\cite{Brakerski2021-bg, Amos2020-ls, ShmueliZ25}. It would be of great interest if such protocols could be adapted to generate randomness in a decentralized setting, which may be useful for other consensus protocols as well. Very recently, it was shown that quantum computers can be used to generate transferable certified min-entropy using only classical communication (though verification requires a quantum computer)~\cite{Casper2026-jx}. If this result could be strengthened from min-entropy to uniform randomness, it could potentially be used to remove the random beacon assumption (at least among participants with quantum computers).
	\end{itemize}
	
	Given the extensive developments in classical consensus in recent years, and the natural suitability of various quantum capabilities and primitives to the requirements of consensus, we believe there could be other fruitful applications of quantum tools to these problems, which will become increasingly relevant as more powerful quantum computers come online. Thus, while on the one hand posing a security risk to existing classical protocols, quantum computers may unlock new possibilities in this space.

	\newpage
	\appendix

    \section{Timing and Spatial Resolution of Position Verification} \label{app:timing_resolution}
    
\begin{figure}[h!]
\centering
\begin{tikzpicture}[x=1.0cm,y=0.95cm,>={Stealth[length=2.2mm]},font=\small]
  \draw[thick] (-4.3,0) -- (4.3,0);
  \node[below] at (-4,-0.12) {$V_0$};  \node[below] at (0,-0.12) {$P$};  \node[below] at (4,-0.12) {$V_1$};
  \fill (-4,0) circle (1.6pt); \fill (0,0) circle (1.6pt); \fill (4,0) circle (1.6pt);
  \draw[dashed] (0,0) -- (0,8.6);
  \draw[->,thick] (-4.9,0.15) -- (-4.9,8.4) node[left,pos=0.5] {$t$};

  \draw[->,thick] (-4,0.25) -- node[above,sloped,pos=0.55] {$k$} (-0.06,1.35);
  \node[right,align=center,font=\footnotesize] at (0.22,2.75)
     {$y, \ket{x_0} + \ket{x_1}$
     \\[1pt] $ \uparrow $
     \\[1pt] $\sum_x \ket{x}\ket{f_k(x)}$
     \\[1pt] $ \uparrow $
     \\[1pt] $\ket{0}$};
  \draw[->,thick] (-0.06,4.15) -- node[below,sloped,pos=0.5] {$y$} (-4,5.05);
  \draw[->,thick] (0.06,4.15) -- node[below,sloped,pos=0.5] {$y$} (4,5.05);

  \draw[->,very thick] (-4,4.75) -- node[above,sloped,pos=0.45] {$b_0$} (-0.06,5.60);
  \draw[->,very thick] (4,4.75) -- node[above,sloped,pos=0.45] {$b_1$} (0.06,5.60);
  \node[right,align=center,font=\footnotesize] at (0.0,6.30)
       {$\,\,z$
     \\[1pt] $\quad \quad\quad \quad \uparrow \mathcal{M}_{b_{0}\oplus b_{1}}$
     \\[1pt] $\left|x_{0}\right\rangle +\left|x_{1}\right\rangle$};
  \draw[->,very thick] (-0.06,7.00) -- node[above,sloped,pos=0.5] {$z$} (-4,7.90);
  \draw[->,very thick] (0.06,7.00) -- node[above,sloped,pos=0.5] {$z$} (4,7.90);

  \draw[decorate,decoration={brace,mirror,amplitude=4pt},thick] (4.6,5.60) -- (4.6,7.00)
     node[midway,right=4pt,font=\small] {$\tau_r$};
  \draw[decorate,decoration={brace,amplitude=4pt},thick] (4.6,4.15) -- (4.6,0.25)
     node[midway,right=4pt,font=\small] {$T_{\mathrm{inst}}$};

\end{tikzpicture}
\caption{Timing of Position Verification with Classical Communication~\cite{Liu2022-ci}. $T_{\mathrm{inst}}$ is the time required to prepare a TCF state by the prover and measure the output register. It does not require strict timing. The timing that sets the ultimate resolution of CVPV is $\tau_r$. This is the time required to conduct a basis change conditioned on a challenge bit $b_0 \oplus b_1$ received by the verifier, which is done by applying single-qubit gates, as well as measurement and decoding.}
\label{fig:timing}
\end{figure}

\begin{table}[h]
\centering
\small
\begin{tabular}{l l l l}
\toprule[.5mm]
& \multicolumn{2}{c}{\textbf{Superconducting}} & \textbf{Neutral atoms} \\
\cmidrule(lr){2-3}
\textbf{Component} & \textbf{Aggressive} & \textbf{Demonstrated} & \\
\midrule
Transversal basis change & $30$--$50$\,ns & $50$\,ns & $0.1$--$2\,\mu$s \cite{Radnaev25} \\
Logical readout & $0.5$--$1\,\mu$s & $1\,\mu$s & $1$--$10$\,ms \cite{Norcia23, Graham23} \\
Real-time decoding \cite{CauneEtAl26, RiverlaneLCD25, Acharya24} & $\sim 5\,\mu$s & \begin{tabular}{@{}l@{}}$64\,\mu$s ($d{=}5$); \\ $9.6\,\mu$s (9 rounds)\end{tabular}   & $\sim 5\,\mu$s \\
\midrule
\textbf{Response time} $\tau_r$ & $\sim 10\,\mu$s & $\sim 70\,\mu$s & $\sim 1$--$10$\,ms \\
\midrule
$\varepsilon = c\tau_r$ & $3$\,km & $21$\,km & $3\times10^2$--$3\times10^3$\,km \\
$\Gamma_{1\mathrm{D}} \approx 2\varepsilon$ (two stations) & $\sim 6$\,km & $\sim 40$\,km & $\sim 6\times10^2$--$6\times10^3$\,km \\
\bottomrule[.5mm]
\end{tabular}
\caption{Time-critical operations required for position verification. The prover's response time is dominated by classical real-time decoding rather than by quantum computation when using superconducting qubits. The long readout times of neutral atom platforms make them unsuitable candidates for implementing position verification, since the resolution of position verification becomes comparable to the radius of the Earth.}\label{tab:timed}
\end{table}
    Position verification relies on precisely-timed messages exchange between the prover and a set of verifiers. In \Cref{fig:timing} we illustrate CVPV in one spatial dimension. The prover prepares a state encoding a TCF indexed by $k$ from some family, the output register is measured and reported to the verifiers, and upon reception of a challenge bit from the verifiers the prover measures the remaining state in one of two bases and reports the measurement result $z$. As shown by \cite{Liu2022-ci}, it is only the later stage (between the sending of challenge bits $b_0,b_1$ and reception of $z$) that is precisely timed, and upon which security of the scheme rests. For any class of proofs of quantumness, this involves applying a transversal one-qubit gate (or no gates at all, based on the value of $b$), measurement, and decoding (since the architecture must be fault-tolerant). Given the overheads associated with this and the fact that any delay corresponds to a reduction in resolution proportional to the distance that light travels in that time, it is natural to ask whether such as scheme is practical. We show in~\Cref{tab:timed} that the scheme is indeed viable when using superconducting qubits, but not when using neutral atoms due to their long readout times. In fact, the bottleneck in the response time (when disregarding the propagation delays themselves which are used to verify the position of the prover) is classical decoding. Security of CVPV is guaranteed when this process is repeated for multiple rounds (roughly linear in $\lambda$), as detailed in \cite{Liu2022-ci}. 

    When CVPV is used within our scheme, there is an additional concern that security can be compromised due to the time required for state preparation, which we denote by $T_{\mathrm{inst}}$. We first note that there are variants where all the required states are prepared before a position is sampled at every reconfiguration, in which case $T_{\mathrm{inst}}$ does not affect security of the protocol at all, though this comes at the cost of larger quantum memory requirements and energy consumption. If instead a position is sampled and then the prover is given time to perform state preparation, a concern is that a malicious prover will use this time to transport a quantum computer, at the expense of failing some rounds of CVPV. Denoting the number of rounds that can be failed by $r_f$, the spatial scale that ultimately sets the resolution at which security is possible is determined by the distance by which it is feasible to transport a quantum computer in time $r_fT_{\mathrm{inst}}$. Preliminary resource estimates indicate that when using efficient proofs of quantumness based on ring-LWE~\cite{Brakerski2020-yo}, $T_{\mathrm{inst}}$ is about one second and this results in a resolution of the scale $\sim 10$km. We leave detailed resource estimates to follow-up work. This resolution is in a regime that minimizes quantum memory. It can be reduced further by increased parallelization. The truly fundamental barrier on resolution is the one set by $\tau_r$ which we discuss above.

    The requirements from the verifiers are well within the capabilities of modern telecommunications hardware. We require direct optical links between verifiers and prover, hence cell towers or satellites can meet the synchronization and data transmission rates required by the protocol. A reasonable implementation could also separate the role of senders and receivers within the set of CVPV verifiers. The senders only need to send a single challenge bit and thus a network of cheap radio frequency receivers will likely suffice. The receivers require higher bandwidth since the size of the response scales with the security parameters of the problem, yet for this purpose a small number of satellites or cell towers can supply the necessary bandwidth and coverage. A back-of-the-envelope calculation gives $10^2$ such receivers in low-Earth orbit, while in geostationary orbit only $2$ suffice, yet achieving high-bandwidth communication may be challenging.

	\section{Solida Consensus Algorithms} \label{app:solida_algs}
	
	For completeness, we reproduce the Steady State and View Change algorithms from Solida~\cite{AbrahamMNRS17} almost verbatim. While Solida has to handle PoW solution events that update the committee leader, we do not, hence our protocol is even simpler and does not require keeping track of the lifespan number of a committee.
	
	We denote by $c$ the committee index (i.e., the number of reconfigurations performed), $s$ the current time slot (ranging from $1$ to $T'$), and by $v$ the view number. We use these to index the committee at every step of the protocol, using the notation $\mathcal{C}(c, v, s)$. Since our committee is ordered, the leader in any reconfiguration step and view, denoted by $L(c,v)$, is defined to be the $v \mod n$-th member of $\mathcal{C}_c$.
	We denote by $\langle x \rangle_{i}$ the message $x$ signed by the $i$-th committee member using their private key (with the corresponding public key known to all participants). $\langle x \rangle_{i}$ correspondingly indicates a signature by the current round leader $L$. When the signer is clear from context, we simply use $\langle x \rangle$.
	
	This follows the standard structure of PBFT protocols, in which two rounds of voting are used to construct a quorum certificate, and unresponsive leaders are replaced as needed. Here $f$ is the number of Byzantine members, and the protocol is secure as long as $f < n/3$.

	\subsection*{4.2 Steady State}
	\begin{itemize}
		\item \textbf{(Propose)} The leader $L$ picks a batch of valid transactions $\{tx\}$ and then broadcasts $\{tx\}$ and $\langle \text{propose}, c, v, s, h \rangle_L$, where $h$ is the hash digest of $\{tx\}$. After receiving $\{tx\}$ and $\langle \text{propose}, c, v, s, h \rangle_L$, a member $M \in \mathcal{C}(c, v, s)$ checks:
		\begin{itemize}
			\item $L = L(c, v)$ and $L$ has not sent a different proposal,
			\item $s$ is a fresh slot.
			\item $\{tx\}$ is a set of valid transactions whose digest is $h$.
		\end{itemize}
		\item \textbf{(Prepare)} If all the above checks pass, $M$ broadcasts $\langle \text{prepare}, c, v, s, h \rangle$. After receiving $2f+1$ matching prepare messages, a member $M \in \mathcal{C}(c, v, s)$ accepts the proposal (represented by its digest $h$), and concatenates the $2f+1$ matching prepare messages into an accept certificate $\mathcal{A}$.
		\item \textbf{(Commit)} Upon accepting $h$, $M$ broadcasts $\langle \text{commit}, c, v, s, h \rangle$. After receiving $2f+1$ matching commit messages, a member $M \in \mathcal{C}(c, v, s)$ commits $\{tx\}$ into slot $s$, and concatenates the $2f+1$ matching commit messages into a commit certificate $\mathcal{Q}$.
	\end{itemize}
	
	There is an additional notification step that propagates the new block to nodes outside the committee:
	\begin{itemize}
		\item \textbf{(Notify)} Upon committing $h$, $M$ sends $\langle \langle \text{notify}, c, v, s, h \rangle, \mathcal{Q} \rangle$ to all other members to notify them about the decision. $M$ also starts propagating this decision on the peer-to-peer network to miners, users, merchants, etc. $M$ then moves to slot $s+1$.
	\end{itemize}
	Upon receiving a notify message like the above, a member commits $h$, sends and propagates its own notify message if it has not already done so, and then moves to slot $s+1$.
	
	\subsection*{4.3 View Change}
	This protocol handles an unresponsive leader by replacing them:
	
	\begin{itemize}
		\item \textbf{(View-change)} Whenever a member $M$ moves to a new slot $s$ in the steady state, it starts a timer $T$. If $T$ reaches $4\Delta$ and $M$ still has not committed slot $s$, then $M$ abandons the current leader and broadcasts $\langle \text{view-change}, c, v \rangle$. 
		
		Upon receiving $2f+1$ matching view-change messages for $(c, v)$, if a member $M$ is not in a view higher than $(c, v)$, it forwards the $2f+1$ view-change messages to the new leader $L' = L(c, v+1)$. After that, if $M$ does not receive a new-view message from $L'$ within $2\Delta$, then $M$ abandons $L'$ and broadcasts $\langle \text{view-change}, c, v + 1 \rangle$.
		
		\item \textbf{(New-view)} Upon receiving $2f+1$ matching view-change messages for $(c, v)$, the new leader $L' = L(c, v+1)$ concatenates them into a view-change certificate $V$, broadcasts $\langle \text{new-view}, c, v+1, \mathcal{V} \rangle_L$ and enters view $(c, v+1)$. Upon receiving a $\langle \text{new-view}, c, v, V \rangle$ message, if a member $M$ is not in a view higher than $(c, v)$, it enters view $(c, v)$ and starts a timer $T$. If $T$ reaches $8\Delta$ and still no new slot is committed, then $M$ abandons $L'$ and broadcasts $\langle \text{view-change}, c, v \rangle$.
		
		\item \textbf{(Status)} Upon entering a new view $(c, v)$, $M$ sends 
		\begin{equation}
			\langle \langle \text{status}, c, v, s-1, h, s, h' \rangle, \mathcal{Q}, \mathcal{A} \rangle
		\end{equation}
		to the new leader $L' = L(c, v)$. In the above message, $h$ is the value committed in $s-1$ and $\mathcal{Q}$ is the corresponding commit certificate; $h'$ is the value accepted by $M$ in slot $s$ and $\mathcal{A}$ is the corresponding accept certificate ($h' = \mathcal{A} = \perp$ if $M$ has not accepted any value for slot $s$). We call the inner part of the message (i.e., excluding $\mathcal{Q}, \mathcal{A}$) its header.
		
		Upon receiving $2f+1$ status messages, $L'$ concatenates the $2f+1$ status headers to form a status certificate $S$. $L'$ then picks a status message that reports the highest last-committed slot $s^*$; if there is a tie, $L'$ picks the one that reports the highest-ranked accepted value in slot $s^*+1$. Let the two certificates in this message be $\mathcal{Q}^*$ and $\mathcal{A}^*$ ($\mathcal{A}^*$ might be $\perp$).
		
		\item \textbf{(Re-propose)} The new leader $L'$ broadcasts 
		\begin{equation}
			\langle \text{repropose}, c, v, s^*+1, h', \mathcal{S}, \mathcal{Q}^*, \mathcal{A}^* \rangle_{L'}.
		\end{equation}
		In the above message, $s^*$ should be the highest last-committed slot reported in $\mathcal{S}$. $h'$ should match the value in $\mathcal{A}$ if $\mathcal{A}^* \neq \perp$; if $\mathcal{A}^* = \perp$ then $L'$ can choose $h'$ freely. %
		The repropose message is invalid if any of these conditions is violated: $s^*$ is not the highest committed slot, $\mathcal{Q}$ is not for slot $s^*$, $\mathcal{A}$ is not for the highest-ranked accepted value for $s^*+1$, or $h'$ is not the value certified by $\mathcal{A}$. 
		
		Upon receiving a valid repropose message, a member $M$ commits slot $s^*$ if it has not already; $M$ then executes the prepare/commit/notify steps as in the steady state for slot $s^*+1$ and marks all slots $> s^*+1$ fresh for view $(c, v)$.
	\end{itemize}
	
	Security of these protocols is proved in~\cite{AbrahamMNRS17}.

    \section{AI Disclosure}

    We used Claude Opus 4.7 and 5 for literature review and high-level discussions about framing the results, which affected the Introduction and Discussion sections. GPT-5.5 was used for routine editorial work. The authors verified the correctness and originality of all content, including references.

    \newpage
    \bibliographystyle{plainurl} %
    \bibliography{paperpile,refs,newrefs}

\begin{thebibliography}{100}

\bibitem{Aaronson2009-uz}
Scott Aaronson.
\newblock Quantum copy-protection and quantum money.
\newblock In {\em 24th Annual IEEE Conference on Computational Complexity},
  pages 229--242. IEEE Computer Soc., Los Alamitos, CA, 2009.

\bibitem{AbrahamMNRS17}
Ittai Abraham, Dahlia Malkhi, Kartik Nayak, Ling Ren, and Alexander Spiegelman.
\newblock Solida: {A} blockchain protocol based on reconfigurable byzantine
  consensus.
\newblock In James Aspnes, Alysson Bessani, Pascal Felber, and Jo{\~{a}}o
  Leit{\~{a}}o, editors, {\em 21st International Conference on Principles of
  Distributed Systems, {OPODIS} 2017, Lisbon, Portugal, December 18-20, 2017},
  LIPIcs, pages 25:1--25:19. Schloss Dagstuhl - Leibniz-Zentrum f{\"{u}}r
  Informatik, 2017.
\newblock URL: \url{https://doi.org/10.4230/LIPIcs.OPODIS.2017.25}, \href
  {https://doi.org/10.4230/LIPICS.OPODIS.2017.25}
  {\path{doi:10.4230/LIPICS.OPODIS.2017.25}}.

\bibitem{Amos2020-ls}
Ryan Amos, Marios Georgiou, Aggelos Kiayias, and Mark Zhandry.
\newblock One-shot signatures and applications to hybrid quantum/classical
  authentication.
\newblock In {\em Proceedings of the 52nd Annual ACM SIGACT Symposium on Theory
  of Computing}, New York, NY, USA, 2020. ACM.

\bibitem{Babbush2021-nu}
Ryan Babbush, Jarrod~R McClean, Michael Newman, Craig Gidney, Sergio Boixo, and
  Hartmut Neven.
\newblock Focus beyond quadratic speedups for error-corrected quantum
  advantage.
\newblock {\em PRX quantum}, 2(1):1, 2021.

\bibitem{Babbush2026-ue}
Ryan Babbush, Adam Zalcman, Craig Gidney, Michael Broughton, Tanuj Khattar,
  Hartmut Neven, Thiago Bergamaschi, Justin Drake, and Dan Boneh.
\newblock Securing elliptic curve cryptocurrencies against quantum
  vulnerabilities: Resource estimates and mitigations, 2026.
\newblock \href {https://arxiv.org/abs/2603.28846} {\path{arXiv:2603.28846}}.

\bibitem{BallGHKP24}
Marshall Ball, Juan~A. Garay, Peter Hall, Aggelos Kiayias, and Giorgos
  Panagiotakos.
\newblock Towards permissionless consensus in the standard model via
  fine-grained complexity.
\newblock In Leonid Reyzin and Douglas Stebila, editors, {\em Advances in
  Cryptology - {CRYPTO} 2024 - 44th Annual International Cryptology Conference,
  Santa Barbara, CA, USA, August 18-22, 2024, Proceedings, Part {II}}, Lecture
  Notes in Computer Science, pages 113--146. Springer, 2024.
\newblock \href {https://doi.org/10.1007/978-3-031-68379-4\_4}
  {\path{doi:10.1007/978-3-031-68379-4\_4}}.

\bibitem{BarZurE25}
Roi Bar-Zur and Ittay Eyal.
\newblock Elastic restaking networks, 2025.
\newblock URL: \url{https://arxiv.org/abs/2503.00170}, \href
  {https://arxiv.org/abs/2503.00170} {\path{arXiv:2503.00170}}.

\bibitem{Becker68}
Gary~S. Becker.
\newblock Crime and punishment: An economic approach.
\newblock {\em Journal of Political Economy}, 76(2):169--217, 1968.
\newblock \href {https://doi.org/10.1086/259394} {\path{doi:10.1086/259394}}.

\bibitem{Ben-David2023-in}
Shalev Ben-David and Or~Sattath.
\newblock Quantum tokens for digital signatures.
\newblock {\em Quantum}, 7:901, 2023.
\newblock \href {https://arxiv.org/abs/1609.09047v8}
  {\path{arXiv:1609.09047v8}}.

\bibitem{Ben-Or2005-ag}
Michael Ben-Or and Avinatan Hassidim.
\newblock Fast quantum byzantine agreement.
\newblock In {\em Proceedings of the thirty-seventh annual ACM symposium on
  Theory of computing}, New York, NY, USA, 2005. ACM.

\bibitem{Bennett1984-xn}
Charles~H Bennett and Gilles Brassard.
\newblock Quantum cryptography: Public key distribution and coin tossing.
\newblock In {\em Proceedings of IEEE International Conference on Computers,
  Systems and Signal Processing}, pages 175--179, 1984.

\bibitem{Bissias2020}
George Bissias.
\newblock Radium: Improving dynamic pow targeting.
\newblock In {\em Data Privacy Management, Cryptocurrencies and Blockchain
  Technology (DPM/CBT @ ESORICS)}, volume 12484 of {\em Lecture Notes in
  Computer Science}, pages 374--389. Springer, 2020.
\newblock \href {https://doi.org/10.1007/978-3-030-66172-4_24}
  {\path{doi:10.1007/978-3-030-66172-4_24}}.

\bibitem{Bissias2018-qj}
George Bissias, Brian~N Levine, and David Thibodeau.
\newblock Using economic risk to model miner hash rate allocation in
  cryptocurrencies.
\newblock In {\em Lecture Notes in Computer Science}, Lecture Notes in Computer
  Science, pages 155--172. Springer International Publishing, Cham, 2018.

\bibitem{Bistarelli2024-oj}
Stefano Bistarelli, Cosimo Laneve, Ivan Mercanti, Adele Veschetti, and
  {Others}.
\newblock Analyzing the fairness of proof of stake ethereum.
\newblock In {\em DLT}, 2024.

\bibitem{BCS22}
Andreas Bluhm, Matthias Christandl, and Florian Speelman.
\newblock A single-qubit position verification protocol that is secure against
  multi-qubit attacks.
\newblock {\em Nature Physics}, 18(6):623--626, 2022.
\newblock \href {https://doi.org/10.1038/s41567-022-01577-0}
  {\path{doi:10.1038/s41567-022-01577-0}}.

\bibitem{Boneh2018-jq}
Dan Boneh, Joseph Bonneau, Benedikt B{\"{u}}nz, and Ben Fisch.
\newblock Verifiable delay functions.
\newblock In {\em Annual international cryptology conference}, pages 757--788,
  2018.

\bibitem{Boneh2011-dm}
Dan Boneh, {\"{O}}zg{\"{u}}r Dagdelen, Marc Fischlin, Anja Lehmann, Christian
  Schaffner, and Mark Zhandry.
\newblock Random oracles in a quantum world.
\newblock In {\em International conference on the theory and application of
  cryptology and information security}, pages 41--69, 2011.

\bibitem{BS23}
Dan Boneh and Victor Shoup.
\newblock A graduate course in applied cryptography.
\newblock {\em Draft 0.6}, 2023.
\newblock URL: \url{https://toc.cryptobook.us/book.pdf}.

\bibitem{Bonneau2022-di}
Joseph Bonneau and Valeria Nikolaenko.
\newblock Public randomness and randomness beacons.
\newblock
  \url{https://a16zcrypto.com/posts/article/public-randomness-and-randomness-beacons/},
  2022.

\bibitem{BorgeKJGGF17}
Maria Borge, Eleftherios Kokoris-Kogias, Philipp Jovanovic, Linus Gasser,
  Nicolas Gailly, and Bryan Ford.
\newblock Proof-of-personhood: Redemocratizing permissionless cryptocurrencies.
\newblock In {\em 2017 {IEEE} European Symposium on Security and Privacy
  Workshops (EuroS\&PW)}, pages 23--26. IEEE, 2017.
\newblock \href {https://doi.org/10.1109/EuroSPW.2017.46}
  {\path{doi:10.1109/EuroSPW.2017.46}}.

\bibitem{Bowman2021-mr}
Mic Bowman, Debajyoti Das, Avradip Mandal, and Hart Montgomery.
\newblock On elapsed time consensus protocols.
\newblock In {\em Lecture Notes in Computer Science}, Lecture notes in computer
  science, pages 559--583. Springer International Publishing, Cham, 2021.

\bibitem{Brakerski2021-bg}
Zvika Brakerski, Paul Christiano, Urmila Mahadev, Umesh Vazirani, and Thomas
  Vidick.
\newblock A cryptographic test of quantumness and certifiable randomness from a
  single quantum device.
\newblock {\em J. ACM}, 68(5):1--47, 2021.

\bibitem{Brakerski2023-ok}
Zvika Brakerski, Alexandru Gheorghiu, Gregory~D Kahanamoku-Meyer, Eitan Porat,
  and Thomas Vidick.
\newblock Simple tests of quantumness also certify qubits.
\newblock In {\em Lecture Notes in Computer Science}, Lecture notes in computer
  science, pages 162--191. Springer Nature Switzerland, Cham, 2023.

\bibitem{Brakerski2020-yo}
Zvika Brakerski, Venkata Koppula, Umesh Vazirani, and Thomas Vidick.
\newblock Simpler proofs of quantumness.
\newblock In {\em 15th Conference on the Theory of Quantum Computation,
  Communication and Cryptography (TQC 2020)}, volume 158 of {\em Leibniz
  International Proceedings in Informatics (LIPIcs)}, pages 8:1--8:14. Schloss
  Dagstuhl--Leibniz-Zentrum f{\"{u}}r Informatik, 2020.

\bibitem{BrakerskiKVV20}
Zvika Brakerski, Venkata Koppula, Umesh~V. Vazirani, and Thomas Vidick.
\newblock Simpler proofs of quantumness.
\newblock In Steven~T. Flammia, editor, {\em 15th Conference on the Theory of
  Quantum Computation, Communication and Cryptography, {TQC} 2020, Riga,
  Latvia, June 9-12, 2020}, LIPIcs, pages 8:1--8:14. Schloss Dagstuhl -
  Leibniz-Zentrum f{\"{u}}r Informatik, 2020.
\newblock URL: \url{https://doi.org/10.4230/LIPIcs.TQC.2020.8}, \href
  {https://doi.org/10.4230/LIPICS.TQC.2020.8}
  {\path{doi:10.4230/LIPICS.TQC.2020.8}}.

\bibitem{Budish2025-fy}
Eric Budish.
\newblock Trust at scale: The economic limits of cryptocurrencies and
  blockchains.
\newblock {\em Q. J. Econ.}, 140(1):1--62, 2025.

\bibitem{Budish2024-sm}
Eric Budish, Andrew Lewis-Pye, and Tim Roughgarden.
\newblock The economic limits of permissionless consensus.
\newblock In {\em Proceedings of the 25th ACM Conference on Economics and
  Computation}, pages 704--731, New York, NY, USA, 2024. ACM.

\bibitem{Buhrman2014-qx}
Harry Buhrman, Nishanth Chandran, Serge Fehr, Ran Gelles, Vipul Goyal, Rafail
  Ostrovsky, and Christian Schaffner.
\newblock Position-based quantum cryptography: Impossibility and constructions.
\newblock {\em SIAM J. Comput.}, 43(1):150--178, 2014.

\bibitem{Buterin2022-yl}
Vitalik Buterin.
\newblock What in the ethereum application ecosystem excites me.
\newblock \url{https://vitalik.eth.limo/general/2022/12/05/excited.html}, 2022.

\bibitem{Buterin2014-nr}
Vitalik Buterin and {Others}.
\newblock A next-generation smart contract and decentralized application
  platform.
\newblock {\em white paper}, 3(37):2--1, 2014.

\bibitem{CCAF-map}
{Cambridge Centre for Alternative Finance}.
\newblock Cambridge blockchain network sustainability index: Bitcoin mining
  map.
\newblock \url{https://ccaf.io/cbnsi/cbeci/mining_map}.
\newblock Accessed September 2026.

\bibitem{CarlstenKWN16}
Miles Carlsten, Harry~A. Kalodner, S.~Matthew Weinberg, and Arvind Narayanan.
\newblock On the instability of bitcoin without the block reward.
\newblock In {\em Proceedings of the 2016 ACM SIGSAC Conference on Computer and
  Communications Security (CCS)}, pages 154--167. ACM, 2016.
\newblock \href {https://doi.org/10.1145/2976749.2978408}
  {\path{doi:10.1145/2976749.2978408}}.

\bibitem{Casper2026-jx}
Ofer Casper, Barak Nehoran, and Or~Sattath.
\newblock Publicly certifiable min-entropy without quantum communication.
\newblock {\em Cryptology ePrint Archive}, 2026.

\bibitem{Castro1999-tp}
Miguel Castro and Barbara Liskov.
\newblock Practical byzantine fault tolerance.
\newblock In {\em Proceedings of the third symposium on Operating systems
  design and implementation}, OSDI '99, pages 173--186, USA, 1999. USENIX
  Association.

\bibitem{CauneEtAl26}
Laura Caune, Luka Skoric, Nick~S. Blunt, Archibald Ruban, Jimmy McDaniel,
  Joseph~A. Valery, Andrew~D. Patterson, Alexander~V. Gramolin, Joonas
  Majaniemi, Kenton~M. Barnes, Tomasz Bialas, Okan Bu{\u g}dayc{\i}, Ophelia
  Crawford, Gy{\"o}rgy~P. Geh{\'e}r, Hari Krovi, Elisha Matekole, Canberk
  Topal, Stefano Poletto, Michael Bryant, Kalan Snyder, Neil~I. Gillespie,
  Glenn Jones, Kauser Johar, Earl~T. Campbell, and Alexander~D. Hill.
\newblock Demonstrating real-time and low-latency quantum error correction with
  superconducting qubits.
\newblock {\em Nature Communications}, 17:7383, 2026.
\newblock \href {https://doi.org/10.1038/s41467-026-73331-6}
  {\path{doi:10.1038/s41467-026-73331-6}}.

\bibitem{ChamanaraGM23}
Sanaz Chamanara, S.~Arman Ghaffarizadeh, and Kaveh Madani.
\newblock The environmental footprint of {Bitcoin} mining across the globe:
  Call for urgent action.
\newblock {\em Earth's Future}, 11(10):e2023EF003871, 2023.
\newblock \href {https://doi.org/10.1029/2023EF003871}
  {\path{doi:10.1029/2023EF003871}}.

\bibitem{Chandran2009-ga}
Nishanth Chandran, Vipul Goyal, Ryan Moriarty, and Rafail Ostrovsky.
\newblock Position based cryptography.
\newblock In {\em Advances in Cryptology - CRYPTO 2009}, Lecture notes in
  computer science, pages 391--407. Springer Berlin Heidelberg, Berlin,
  Heidelberg, 2009.

\bibitem{Cohen2019-pn}
B~Cohen and Krzysztof Pietrzak.
\newblock The chia network blockchain, 2019.

\bibitem{CS20}
Andrea Coladangelo and Or~Sattath.
\newblock A quantum money solution to the blockchain scalability problem.
\newblock {\em Quantum}, 4:297, 2020.
\newblock URL: \url{https://doi.org/10.22331/q-2020-07-16-297}, \href
  {https://doi.org/10.22331/Q-2020-07-16-297}
  {\path{doi:10.22331/Q-2020-07-16-297}}.

\bibitem{CCRI22}
{Crypto Carbon Ratings Institute}.
\newblock The merge -- implications on the electricity consumption and carbon
  footprint of the {Ethereum} network.
\newblock Technical report, Crypto Carbon Ratings Institute ({CCRI}), September
  2022.
\newblock URL: \url{https://carbon-ratings.com/dl/eth-report-2022}.

\bibitem{DeVriesGKS22}
Alex de~Vries, Ulrich Gallersd{\"o}rfer, Lena Klaa{\ss}en, and Christian Stoll.
\newblock Revisiting {Bitcoin}'s carbon footprint.
\newblock {\em Joule}, 6(3):498--502, 2022.
\newblock \href {https://doi.org/10.1016/j.joule.2022.02.005}
  {\path{doi:10.1016/j.joule.2022.02.005}}.

\bibitem{Dembo2020-vd}
Amir Dembo, Sreeram Kannan, Ertem~Nusret Tas, David Tse, Pramod Viswanath,
  Xuechao Wang, and Ofer Zeitouni.
\newblock Everything is a race and nakamoto always wins.
\newblock In {\em Proceedings of the 2020 ACM SIGSAC Conference on Computer and
  Communications Security}, New York, NY, USA, 2020. ACM.

\bibitem{Demers1987-rr}
Alan Demers, Dan Greene, Carl Hauser, Wes Irish, and John Larson.
\newblock Epidemic algorithms for replicated database maintenance.
\newblock In {\em Proceedings of the sixth annual ACM Symposium on Principles
  of distributed computing - PODC '87}, New York, New York, USA, 1987. ACM
  Press.

\bibitem{Dolev1983-yg}
Danny Dolev and H~Raymond Strong.
\newblock Authenticated algorithms for byzantine agreement.
\newblock {\em SIAM Journal on Computing}, 12(4):656--666, 1983.

\bibitem{Drake2023-wu}
Justin Drake.
\newblock One-shot signatures.
\newblock Talk given at the Programmable Cryptography Conference (ProgCrypto),
  2023.

\bibitem{Dwork1988-dh}
Cynthia Dwork, Nancy Lynch, and Larry Stockmeyer.
\newblock Consensus in the presence of partial synchrony.
\newblock {\em J. ACM}, 35(2):288--323, 1988.

\bibitem{DN93}
Cynthia Dwork and Moni Naor.
\newblock Pricing via processing or combatting junk mail.
\newblock In {\em Advances in Cryptology---CRYPTO '92}, pages 139--147.
  Springer, 1993.
\newblock \href {https://doi.org/10.1007/3-540-48071-4_10}
  {\path{doi:10.1007/3-540-48071-4_10}}.

\bibitem{Dziembowski2015-cu}
Stefan Dziembowski, Sebastian Faust, Vladimir Kolmogorov, and Krzysztof
  Pietrzak.
\newblock Proofs of space.
\newblock In {\em Annual Cryptology Conference}, pages 585--605, 2015.

\bibitem{Fanti2019-nj}
Giulia Fanti, Leonid Kogan, Sewoong Oh, Kathleen Ruan, Pramod Viswanath, and
  Gerui Wang.
\newblock Compounding of wealth in proof-of-stake cryptocurrencies.
\newblock In {\em International conference on financial cryptography and data
  security}, pages 42--61, 2019.

\bibitem{Gavinsky12}
Dmitry Gavinsky.
\newblock Quantum money with classical verification.
\newblock In {\em Proceedings of the 27th Conference on Computational
  Complexity, {CCC} 2012, Porto, Portugal, June 26-29, 2012}, pages 42--52.
  {IEEE} Computer Society, 2012.
\newblock \href {https://doi.org/10.1109/CCC.2012.10}
  {\path{doi:10.1109/CCC.2012.10}}.

\bibitem{Gavinsky2025-kb}
Dmytro Gavinsky, Dar Gilboa, Siddhartha Jain, Dmitri Maslov, and Jarrod~R
  McClean.
\newblock Anonymous quantum tokens with classical verification, 2025.
\newblock \href {https://arxiv.org/abs/2510.06212} {\path{arXiv:2510.06212}}.

\bibitem{Gencer2018-mo}
Adem~Efe Gencer, Soumya Basu, Ittay Eyal, Robbert van Renesse, and
  Emin~G{\"{u}}n Sirer.
\newblock Decentralization in bitcoin and ethereum networks.
\newblock In {\em Financial Cryptography and Data Security}, Lecture Notes in
  Computer Science, pages 439--457. Springer Berlin Heidelberg, Berlin,
  Heidelberg, 2018.

\bibitem{Gidney2025-dp}
Craig Gidney.
\newblock How to factor 2048 bit {RSA} integers with less than a million noisy
  qubits, 2025.
\newblock \href {https://arxiv.org/abs/2505.15917} {\path{arXiv:2505.15917}}.

\bibitem{Gidney2021-je}
Craig Gidney and Martin Eker\aa{}.
\newblock How to factor 2048 bit {RSA} integers in 8 hours using 20 million
  noisy qubits.
\newblock {\em Quantum}, 5:433, 2021.
\newblock \href {https://arxiv.org/abs/1905.09749v3}
  {\path{arXiv:1905.09749v3}}.

\bibitem{Gidney2024-ab}
Craig Gidney, Noah Shutty, and Cody Jones.
\newblock Magic state cultivation: growing {T} states as cheap as {CNOT} gates,
  2024.
\newblock \href {https://arxiv.org/abs/2409.17595} {\path{arXiv:2409.17595}}.

\bibitem{Girish2026-so}
Uma Girish, Greg Gluch, Shafi Goldwasser, Tal Malkin, Leo Orshansky, and Henry
  Yuen.
\newblock Private proofs of when and where, 2026.
\newblock \href {https://arxiv.org/abs/2601.18961} {\path{arXiv:2601.18961}}.

\bibitem{Google-Quantum-AI-and-Collaborators2025-nq}
{Google Quantum AI and Collaborators}.
\newblock Observation of constructive interference at the edge of quantum
  ergodicity.
\newblock {\em Nature}, 646(8086):825--830, 2025.

\bibitem{Acharya24}
{Google Quantum AI and Collaborators}.
\newblock Quantum error correction below the surface code threshold.
\newblock {\em Nature}, 638:920--926, 2025.
\newblock \href{https://arxiv.org/abs/2408.13687}{arXiv:2408.13687}.

\bibitem{Graham23}
T.~M. Graham, L.~Phuttitarn, R.~Chinnarasu, Y.~Song, et~al.
\newblock Midcircuit measurements on a single-species neutral alkali atom
  quantum processor.
\newblock {\em Physical Review X}, 13:041051, 2023.

\bibitem{ZRW25}
Ziyi Guan, Artur Riazanov, and Weiqiang Yuan.
\newblock Breaking verifiable delay functions in the random oracle model.
\newblock In {\em Advances in cryptology---{CRYPTO} 2025. {P}art {VII}}, volume
  16006 of {\em Lecture Notes in Comput. Sci.}, pages 161--191. Springer, 2025.

\bibitem{HiraharaG21}
Shuichi Hirahara and Fran{\c{c}}ois~Le Gall.
\newblock Test of quantumness with small-depth quantum circuits.
\newblock In Filippo Bonchi and Simon~J. Puglisi, editors, {\em 46th
  International Symposium on Mathematical Foundations of Computer Science,
  {MFCS} 2021, Tallinn, Estonia, August 23-27, 2021}, LIPIcs, pages
  59:1--59:15. Schloss Dagstuhl - Leibniz-Zentrum f{\"{u}}r Informatik, 2021.
\newblock URL: \url{https://doi.org/10.4230/LIPIcs.MFCS.2021.59}, \href
  {https://doi.org/10.4230/LIPICS.MFCS.2021.59}
  {\path{doi:10.4230/LIPICS.MFCS.2021.59}}.

\bibitem{Jelasity2007-zy}
M\'{a}rk Jelasity, Spyros Voulgaris, Rachid Guerraoui, Anne-Marie Kermarrec,
  and Maarten van Steen.
\newblock Gossip-based peer sampling.
\newblock {\em ACM Trans. Comput. Syst.}, 25(3):8, 2007.

\bibitem{JudmayerZSVW17}
Aljosha Judmayer, Alexei Zamyatin, Nicholas Stifter, Artemios~G. Voyiatzis, and
  Edgar~R. Weippl.
\newblock Merged mining: Curse or cure?
\newblock In {\em Data Privacy Management, Cryptocurrencies and Blockchain
  Technology (DPM/CBT @ ESORICS)}, volume 10436 of {\em Lecture Notes in
  Computer Science}, pages 316--333. Springer, 2017.
\newblock \href {https://doi.org/10.1007/978-3-319-67816-0_18}
  {\path{doi:10.1007/978-3-319-67816-0_18}}.

\bibitem{Kahanamoku-Meyer2022-ta}
Gregory~D Kahanamoku-Meyer, Soonwon Choi, Umesh~V Vazirani, and Norman~Y Yao.
\newblock Classically verifiable quantum advantage from a computational bell
  test.
\newblock {\em Nat. Phys.}, 18(8):918--924, 2022.

\bibitem{Kavuri2026-ba}
Gautam~A Kavuri, Yanbao Zhang, Abigail~R Gookin, Soumyadip Patra, Joshua~C
  Bienfang, Honghao Fu, Yusuf Alnawakhtha, Dileep~V Reddy, Michael~D Mazurek,
  Carlos Abell\'{a}n, Waldimar Amaya, Morgan~W Mitchell, Sae~Woo Nam, Carl~A
  Miller, Richard~P Mirin, Martin~J Stevens, Scott Glancy, Emanuel Knill, and
  Lynden~K Shalm.
\newblock Quantum position verification with remote untrusted devices, 2026.
\newblock \href {https://arxiv.org/abs/2601.16892} {\path{arXiv:2601.16892}}.

\bibitem{Kiayias2017-is}
Aggelos Kiayias, Alexander Russell, Bernardo David, and Roman Oliynykov.
\newblock Ouroboros: A provably secure proof-of-stake blockchain protocol.
\newblock In {\em Annual International Cryptology Conference (CRYPTO)}, pages
  357--388, 2017.

\bibitem{King2025-lq}
Robbie King, Robin Kothari, Ryan Babbush, Sergio Boixo, Kostyantyn Kechedzhi,
  Thomas~E O'Brien, and Vadim Smelyanskiy.
\newblock A simplified version of the quantum {OTOC}$^{(2)}$ problem, 2025.
\newblock \href {https://arxiv.org/abs/2510.19751} {\path{arXiv:2510.19751}}.

\bibitem{KokorisKogiasJ16}
Eleftherios Kokoris{-}Kogias, Philipp Jovanovic, Nicolas Gailly, Ismail Khoffi,
  Linus Gasser, and Bryan Ford.
\newblock Enhancing bitcoin security and performance with strong consistency
  via collective signing.
\newblock In Thorsten Holz and Stefan Savage, editors, {\em 25th {USENIX}
  Security Symposium, {USENIX} Security 16, Austin, TX, USA, August 10-12,
  2016}, pages 279--296. {USENIX} Association, 2016.
\newblock URL:
  \url{https://www.usenix.org/conference/usenixsecurity16/technical-sessions/presentation/kogias}.

\bibitem{Kwon2019}
Yujin Kwon, Hyoungshick Kim, Jinwoo Shin, and Yongdae Kim.
\newblock Bitcoin vs. bitcoin cash: Coexistence or downfall of bitcoin cash?
\newblock In {\em 2019 IEEE Symposium on Security and Privacy (SP)}, pages
  935--951. IEEE, 2019.
\newblock \href {https://doi.org/10.1109/SP.2019.00075}
  {\path{doi:10.1109/SP.2019.00075}}.

\bibitem{Lauk2020-du}
Nikolai Lauk, Neil Sinclair, Shabir Barzanjeh, Jacob~P Covey, Mark Saffman,
  Maria Spiropulu, and Christoph Simon.
\newblock Perspectives on quantum transduction.
\newblock {\em Quantum Sci. Technol.}, 5(2):020501, 2020.

\bibitem{PhysRevA.109.012610}
Laura Lewis, Daiwei Zhu, Alexandru Gheorghiu, Crystal Noel, Or~Katz, Bahaa
  Harraz, Qingfeng Wang, Andrew Risinger, Lei Feng, Debopriyo Biswas, Laird
  Egan, Thomas Vidick, Marko Cetina, and Christopher Monroe.
\newblock Experimental implementation of an efficient test of quantumness.
\newblock {\em Phys. Rev. A}, 109:012610, Jan 2024.
\newblock URL: \url{https://link.aps.org/doi/10.1103/PhysRevA.109.012610},
  \href {https://doi.org/10.1103/PhysRevA.109.012610}
  {\path{doi:10.1103/PhysRevA.109.012610}}.

\bibitem{Lewis-Pye2023-zp}
Andrew Lewis-Pye and Tim Roughgarden.
\newblock Permissionless consensus, 2023.
\newblock \href {https://arxiv.org/abs/2304.14701} {\path{arXiv:2304.14701}}.

\bibitem{Liu2022-ci}
Jiahui Liu, Qipeng Liu, and Luowen Qian.
\newblock Beating classical impossibility of position verification, 2022.

\bibitem{Lutz22}
Sander Lutz.
\newblock Is {Solana} decentralized? cloud provider {Hetzner} ban raises
  questions.
\newblock Decrypt,
  \url{https://decrypt.co/113429/is-solana-decentralized-cloud-provider-hetzner-ban-raises-questions},
  November 2022.

\bibitem{Mahadev2022-ts}
Urmila Mahadev.
\newblock Classical verification of quantum computations.
\newblock {\em SIAM J. Comput.}, 51(4):1172--1229, 2022.

\bibitem{Nakamoto2008-bc}
Satoshi Nakamoto.
\newblock Bitcoin: A peer-to-peer electronic cash system, 2008.

\bibitem{NarayananBFMG16}
Arvind Narayanan, Joseph Bonneau, Edward Felten, Andrew Miller, and Steven
  Goldfeder.
\newblock {\em Bitcoin and Cryptocurrency Technologies: A Comprehensive
  Introduction}.
\newblock Princeton University Press, 2016.

\bibitem{NodaOH26}
Shunya Noda, Kyohei Okumura, and Yoshinori Hashimoto.
\newblock An economic analysis of difficulty adjustment algorithms in
  proof-of-work blockchain systems.
\newblock {\em International Economic Review}, 67(1):259--285, 2026.
\newblock \href {https://doi.org/10.1111/iere.70028}
  {\path{doi:10.1111/iere.70028}}.

\bibitem{Norcia23}
M.~A. Norcia, W.~B. Cairncross, K.~Barnes, P.~Battaglino, et~al.
\newblock Midcircuit qubit measurement and rearrangement in a $^{171}$yb atomic
  array.
\newblock {\em Physical Review X}, 13:041034, 2023.

\bibitem{PassS17}
Rafael Pass and Elaine Shi.
\newblock Hybrid consensus: Efficient consensus in the permissionless model.
\newblock In Andr{\'{e}}a~W. Richa, editor, {\em 31st International Symposium
  on Distributed Computing, {DISC} 2017, Vienna, Austria, October 16-20, 2017},
  LIPIcs, pages 39:1--39:16. Schloss Dagstuhl - Leibniz-Zentrum f{\"{u}}r
  Informatik, 2017.
\newblock URL: \url{https://doi.org/10.4230/LIPIcs.DISC.2017.39}, \href
  {https://doi.org/10.4230/LIPICS.DISC.2017.39}
  {\path{doi:10.4230/LIPICS.DISC.2017.39}}.

\bibitem{PratW21}
Julien Prat and Benjamin Walter.
\newblock An equilibrium model of the market for bitcoin mining.
\newblock {\em Journal of Political Economy}, 129(8):2415--2452, 2021.
\newblock \href {https://doi.org/10.1086/714445} {\path{doi:10.1086/714445}}.

\bibitem{Radnaev25}
A.~Radnaev, W.~Chung, D.~Cole, D.~Mason, et~al.
\newblock Universal neutral-atom quantum computer with individual optical
  addressing and nondestructive readout.
\newblock {\em PRX Quantum}, 6:030334, 2025.

\bibitem{Regev2010-zq}
Oded Regev.
\newblock The learning with errors problem (invited survey).
\newblock In {\em 2010 IEEE 25th Annual Conference on Computational
  Complexity}. IEEE, 2010.

\bibitem{Regev2025-ho}
Oded Regev.
\newblock An efficient quantum factoring algorithm.
\newblock {\em J. ACM}, 72(1):1--13, 2025.

\bibitem{Reuters22}
{Reuters}.
\newblock Bitcoin network power slumps as {Kazakhstan} crackdown hits crypto
  miners.
\newblock
  \url{https://www.nbcnews.com/tech/tech-news/bitcoin-network-power-slumps-kazakhstan-crackdown-hits-crypto-miners-rcna11221},
  January 2022.

\bibitem{Roeschlin2026-xw}
Marc Roeschlin, Evangelos Markakis, Raghav Bhaskar, and Aggelos Kiayias.
\newblock Incentivizing geographic diversity for decentralized systems.
\newblock {\em Cryptology ePrint Archive}, 2026.

\bibitem{Roughgarden2024-wn}
Tim Roughgarden.
\newblock The computer in the sky (keynote).
\newblock In {\em Proceedings of the 56th Annual ACM Symposium on Theory of
  Computing}, pages 1--1, New York, NY, USA, 2024. ACM.

\bibitem{Sai2021-uc}
Ashish~Rajendra Sai, Jim Buckley, Brian Fitzgerald, and Andrew~Le Gear.
\newblock Taxonomy of centralization in public blockchain systems: A systematic
  literature review.
\newblock {\em Inf. Process. Manag.}, 58(4):102584, 2021.

\bibitem{Sattath2023-uc}
Or~Sattath.
\newblock Uncloneable cryptography.
\newblock {\em Communications of the ACM}, 66(11):78--86, 2023.
\newblock \href {https://doi.org/10.1145/3576897} {\path{doi:10.1145/3576897}}.

\bibitem{Sedlmeir2020-qi}
Johannes Sedlmeir, Hans~Ulrich Buhl, Gilbert Fridgen, and Robert Keller.
\newblock The energy consumption of blockchain technology: Beyond myth.
\newblock {\em Business and Information Systems Engineering}, 62:599--608,
  2020.

\bibitem{Shi2020-aw}
Elaine Shi.
\newblock Foundations of distributed consensus and blockchains.
\newblock {\em Book manuscript}, 2020.

\bibitem{ShmueliZ25}
Omri Shmueli and Mark Zhandry.
\newblock On one-shot signatures, quantum vs. classical binding, and
  obfuscating permutations.
\newblock In Yael~Tauman Kalai and Seny~F. Kamara, editors, {\em Advances in
  Cryptology - {CRYPTO} 2025 - 45th Annual International Cryptology Conference,
  Santa Barbara, CA, USA, August 17-21, 2025, Proceedings, Part {II}}, Lecture
  Notes in Computer Science, pages 350--383. Springer, 2025.
\newblock \href {https://doi.org/10.1007/978-3-032-01878-6\_12}
  {\path{doi:10.1007/978-3-032-01878-6\_12}}.

\bibitem{Shmueli2025-oh}
Omri Shmueli and Mark Zhandry.
\newblock Unclonable cryptography in linear quantum memory.
\newblock {\em Cryptology ePrint Archive}, 2025.

\bibitem{Shor1994-hb}
P~W Shor.
\newblock Algorithms for quantum computation: discrete logarithms and
  factoring.
\newblock In {\em Proceedings 35th Annual Symposium on Foundations of Computer
  Science}. IEEE Comput. Soc. Press, 1994.

\bibitem{Sompolinsky2015-bk}
Yonatan Sompolinsky and Aviv Zohar.
\newblock Secure high-rate transaction processing in bitcoin.
\newblock In {\em Financial Cryptography and Data Security}, Lecture Notes in
  Computer Science, pages 507--527. Springer Berlin Heidelberg, Berlin,
  Heidelberg, 2015.

\bibitem{StollKG19}
Christian Stoll, Lena Klaa{\ss}en, and Ulrich Gallersd{\"o}rfer.
\newblock The carbon footprint of {Bitcoin}.
\newblock {\em Joule}, 3(7):1647--1661, 2019.
\newblock \href {https://doi.org/10.1016/j.joule.2019.05.012}
  {\path{doi:10.1016/j.joule.2019.05.012}}.

\bibitem{TsabaryE18}
Itay Tsabary and Ittay Eyal.
\newblock The gap game.
\newblock In {\em Proceedings of the 2018 ACM SIGSAC Conference on Computer and
  Communications Security (CCS)}, pages 713--728. ACM, 2018.
\newblock \href {https://doi.org/10.1145/3243734.3243737}
  {\path{doi:10.1145/3243734.3243737}}.

\bibitem{Tully22}
Shawn Tully.
\newblock Bitcoin mining is being banned in countries across the globe and
  threatening the future of crypto.
\newblock Fortune,
  \url{https://fortune.com/2022/01/05/crypto-blackouts-bitcoin-mining-bans-kosovo-iran-kazakhstan-iceland/},
  January 2022.

\bibitem{Unruh2014-eb}
Dominique Unruh.
\newblock Quantum position verification in the random oracle model.
\newblock In {\em Advances in Cryptology -- CRYPTO 2014}, Lecture notes in
  computer science, pages 1--18. Springer Berlin Heidelberg, Berlin,
  Heidelberg, 2014.

\bibitem{Vidick2019-qm}
Thomas Vidick.
\newblock Verifying quantum computations at scale: A cryptographic leash on
  quantum devices.
\newblock {\em Bull. New Ser. Am. Math. Soc.}, 57(1):39--76, 2019.

\bibitem{WahrstatterEYZQ24}
Anton Wahrst{\"{a}}tter, Jens Ernstberger, Aviv Yaish, Liyi Zhou, Kaihua Qin,
  Taro Tsuchiya, Sebastian Steinhorst, Davor Svetinovic, Nicolas Christin,
  Mikolaj Barczentewicz, and Arthur Gervais.
\newblock Blockchain censorship.
\newblock In {\em Proceedings of the {ACM} on Web Conference 2024, {WWW} 2024,
  Singapore, May 13-17, 2024}, pages 1632--1643. {ACM}, 2024.
\newblock \href {https://doi.org/10.1145/3589334.3645431}
  {\path{doi:10.1145/3589334.3645431}}.

\bibitem{Wiesner1983-lt}
Stephen Wiesner.
\newblock Conjugate coding.
\newblock {\em SIGACT News}, 15(1):78--88, 1983.

\bibitem{WZ82}
W.K. Wootters and W.H. Zurek.
\newblock A single quantum cannot be cloned.
\newblock {\em Nature}, 299(5886):802--803, 1982.
\newblock \href {https://doi.org/10.1038/299802a0}
  {\path{doi:10.1038/299802a0}}.

\bibitem{Yakovenko2018-ot}
Anatoly Yakovenko.
\newblock Solana: A new architecture for a high performance blockchain
  {v0}.8.13.
\newblock Technical report, Solana Labs, 2018.

\bibitem{Zhandry2019-ru}
Mark Zhandry.
\newblock Quantum lightning never strikes the same state twice.
\newblock In {\em Advances in Cryptology -- EUROCRYPT 2019}, Lecture notes in
  computer science, pages 408--438. Springer International Publishing, Cham,
  2019.

\bibitem{Zhang2020-jx}
Rong Zhang and Wai~Kin Chan.
\newblock Evaluation of energy consumption in block-chains with proof of work
  and proof of stake.
\newblock In {\em Journal of physics: Conference series}, volume 1584, page
  012023, 2020.

\bibitem{Zhu2021-co}
Daiwei Zhu, Gregory~D Kahanamoku-Meyer, Laura Lewis, Crystal Noel, Or~Katz,
  Bahaa Harraz, Qingfeng Wang, Andrew Risinger, Lei Feng, Debopriyo Biswas,
  Laird Egan, Alexandru Gheorghiu, Yunseong Nam, Thomas Vidick, Umesh Vazirani,
  Norman~Y Yao, Marko Cetina, and Christopher Monroe.
\newblock Interactive protocols for classically-verifiable quantum advantage,
  2021.
\newblock \href {https://arxiv.org/abs/2112.05156} {\path{arXiv:2112.05156}}.

\bibitem{RiverlaneLCD25}
Abbas~B. Ziad, Ankit Zalawadiya, Canberk Topal, Joan Camps, Gy{\"o}rgy~P.
  Geh{\'e}r, Matthew~P. Stafford, and Mark~L. Turner.
\newblock Local clustering decoder as a fast and adaptive hardware decoder for
  the surface code.
\newblock {\em Nature Communications}, 16:11048, 2025.
\newblock \href {https://doi.org/10.1038/s41467-025-66773-x}
  {\path{doi:10.1038/s41467-025-66773-x}}.

\end{thebibliography}

\end{document}